\documentclass[11pt]{article}
\usepackage[margin=1in]{geometry}
\usepackage{fullpage}
\usepackage{tgtermes}
\usepackage[T1]{fontenc}
\usepackage[colorlinks,citecolor=blue,linkcolor=blue,urlcolor=black]{hyperref}
\usepackage{graphicx}
\usepackage{amsfonts,amsmath,amsthm,amssymb,dsfont}
\usepackage{xfrac,nicefrac}
\usepackage{mathdots}
\usepackage{bm,bbm}
\usepackage{url}
\usepackage{paralist}
\usepackage{enumerate}
\usepackage[normalem]{ulem}
\usepackage{xspace}
\xspaceaddexceptions{]\}}
\usepackage{cleveref}
\usepackage{comment}
\usepackage{tabu}
\usepackage{framed}
\usepackage{float,wrapfig}
\usepackage{subfig}
\usepackage{tikz}
\usepackage[usenames,dvipsnames]{pstricks} 

\usepackage{enumitem}
\usepackage{adjustbox}
\usepackage{algorithm}
\usepackage{algpseudocode}

\algnewcommand{\algorithmicand}{\textbf{AND }}
\algnewcommand{\algorithmicor}{\textbf{OR }}
\algnewcommand{\algorithmicxor}{\textbf{XOR }}
\algnewcommand{\algorithmicnot}{\textbf{NOT }}
\algnewcommand{\OR}{\algorithmicor}
\algnewcommand{\AND}{\algorithmicand}
\algnewcommand{\XOR}{\algorithmicxor}
\algnewcommand{\NOT}{\algorithmicnot}
\algnewcommand{\var}{\texttt}

\usetikzlibrary{hobby}
\usetikzlibrary{decorations.markings}
\usetikzlibrary{quotes,angles}

\theoremstyle{plain}

\newtheorem{theorem}{Theorem}[section]
\newtheorem{lemma}[theorem]{Lemma}

\newtheorem{cor}[theorem]{Corollary}

\theoremstyle{definition}
\newtheorem{definition}[theorem]{Definition}

\renewcommand{\epsilon}{\varepsilon}

\title{Shortest non-separating st-path on chordal graphs}

\author{
  Xiao Mao \\
  Massachusetts Institute of Technology \\
  \texttt{xiao\_mao@mit.edu} \\
}

\begin{document}

\setcounter{page}{0} \clearpage
\maketitle
\thispagestyle{empty}

\begin{abstract}
    Many NP-Hard problems on general graphs, such as maximum independence set, maximal cliques and graph coloring can be solved efficiently on chordal graphs. In this paper, we explore the problem of non-separating st-paths defined on edges: for a connected undirected graph and two vertices, a non-separating path is a path between the two vertices such that if we remove all the edges on the path, the graph remains connected. We show that on general graphs, checking the existence of non-separating st-paths is NP-Hard, but the same problem can be solved in linear time on chordal graphs. In the case that such path exists, we introduce an algorithm that finds the shortest non-separating st-path on a connected chordal graph of $n$ vertices and $m$ edges with positive edge lengths that runs in $O(n\log{n} + m)$ time.
\end{abstract}

\newpage


\tikzset{->-/.style={decoration={
  markings,
  mark=at position .5 with {\arrow{>}}},postaction={decorate}}}

\section{Introduction}

\subsection{Overview}

\subsubsection{Non-separating defined on edges vs. vertices}

In this paper, we study non-separating paths defined in the following way:

\begin{definition} [Non-separating path (edge)]
    Let $G = \langle V, E\rangle$ be an undirected graph. A non-separating path is a path $p$ in $G$ such that after removing all edges on $p$. The remaining graph remains connected.
\end{definition}

Traditionally, there have been studies on non-separating paths defined with vertex removal instead of edge removal. 

\begin{definition} [Non-separating path (vertex)]
    Let $G = \langle V, E\rangle$ be an undirected graph. A non-separating path is a path $p$ in $G$ such that after removing all \textit{vertices} on $p$. The remaining graph remains connected.
\end{definition}

For this definition, many studies on their relation to graph connectivity have been done \cite{bollobas, chen, tutte, 4connected, kawarabayashi} since it first gained major academic interest in 1975, when Lovász made a famous conjecture relating path removal to k-connectivity \cite{lovasz}. The first major work on related optimization problem was published in 2009 by B. Y. Wu and S. C. Chen, who showed that the optimization problem on general graphs is NP-hard in the strong sense \cite{nphard}. Later in 2014, Wu showed an efficient $O(N \log {N})$ time algorithm for the optimization problem on grid graphs \cite{wu}. 

In this paper we study the version defined on edges. We are only interested in the optimization problem. Like Wu's work in 2009 and 2014 which proved NP-hardness in general and showed an efficient algorithm on a special type of graphs, we will obtain an efficient algorithm for both the decision and the optimization problem on connected \textit{chordal graphs} with positive edge lengths, and show that the decision problem with non-separating paths defined with edge removal is NP-hard on general graphs. For the rest of the paper, the term ``non-separating path'' refers to the edge version.

\subsubsection{Chordal Graphs}

\begin{definition} [Chordal Graph]
A chordal graph is a simple, undirected graph where all cycles with four or more vertices have at least one chords. A chord is an edge that is not part of the cycle but connects two vertices of the cycle. 
\end{definition}

Given a connected chordal graph $G = \langle V, E\rangle$ with positive lengths on edges and a pair of distinct vertices $S$ and $T$, the decision problem is to decide whether a non-separating path exists between $S$ and $T$, and the optimization problem, which is the main topic of this paper, is to calculate the shortest non-separating path between $S$ and $T$ if such paths exist.

Chordal graphs are a wide class of graphs with many subclasses, including interval graphs, split graphs, Ptolemaic graphs, and K-trees. Many studies have been done on chordal graphs in the past. The most famous technique for these graphs is the computation of perfect elimination ordering, which can be done in linear time by the famous algorithm by Rose, Leuker and Tarjan \cite{perfectelimination}.  It can be used to compute the clique tree of a chordal graph \cite{blair}. The tree decomposition technique is another famous technique, which is due to Gavril's work showing that chordal graphs are intersection graphs of subtrees, and a representation of a chordal graph as an intersection of subtrees forms a tree decomposition of the graph \cite{gavril}.

A significance of chordal graphs is that on these types of graphs, many generally NP-hard problems can be efficiently solved. For example, the perfect elimination ordering can be used to compute the maximum independent set, maximal cliques, graph coloring, etc., which are all examples of NP-hard problems on general graphs. The tree decomoposition technique can be used to solve the routing algorithm for chordal graphs \cite{kosowski}. \textbf{As our work shows, the problem we study is another example of a generally NP-hard problem that can be efficiently solved on chordal graphs}. Our algorithm is not based on famous techniques such as perfect elimination ordering, clique tree or tree decomposition. Instead, it is based on a number of new observations and original sub-procedures related to our problem.

\subsection{Our Results}

Assume that $G$ is a connected chordal graph, and $S, T$ are distinct vertices in $G$. An edge is called a \textit{bridge} if after removing such edge, $G$ becomes disconnected. 

Our main result shows that the optimization problem can be solved efficiently, and the majority of this paper deals with this optimization problem.

\begin{theorem} \label{theomaster}
    For a connected chordal graph $G = \langle V, E\rangle$ with non-negative weights on edges and a pair of vertices $S$ and $T$, the shortest non-separating path between the $S$ and $T$ can be calculated in $O(T_{\textrm{SPTD}}(2|V|, 8|E|))$ time, where $T_{\textrm{SPTD}}(2|V|, 8|E|)$ is the asymptotic time to compute a single source shortest path tree in a \textbf{directed} graph with no more than $2|V|$ vertices and $8|E|$ edges.
\end{theorem}

Since the single source shortest path tree in a directed graph with $n$ vertices and $m$ edges can be calculated in $O(n\log{n} + m)$ time using the famous Dijkstra's algorithm equipped with a Fibonacci heap \cite{dijkstra, tarjanheap}, our problem can be solved in $O(|V|\log{|V|} + |E|)$ time.

For the decision problem, we have the following result:

\begin{theorem} \label{existence}
    There exists a non-separating path from $S$ to $T$ if and only if $S$ and $T$ are not separated by a bridge. When such path exists, the path from $S$ to $T$ that contains the minimum number of edges is a non-separating path.
\end{theorem}

This implies an $O(|V| + |E|)$ time algorithm for checking the existence of non-separating path by doing a simple breadth first graph traversal. Proof for Theorem \ref{existence} will be given in section \ref{existenceproof}.

We also show that the problem of deciding the existence of non-separating st-paths is NP-hard on general graphs. 

\begin{theorem} \label{nphard}
    On general graphs, deciding if non-separating st-paths exist is NP-Hard.
\end{theorem}

\subsection{Paper Organization}

In this paper, we will first give give our algorithm for the optimization problem on chordal graphs, where we will first give an overview of the algorithm, and then prove important key lemma and theorems that we use. With these we can complete the proof for Theorem \ref{existence}. After that, we give details on the components of the main algorithm, and then give our final correctness proof. Finally we will show the NP-hardness of the existence problem on general graphs. 

\section{Definitions and Notations}

Let $G = \langle V, E\rangle$ be a simple undirected chordal graph with edge lengths. For two vertices $u, v \in V$, $u$ and $v$ are \textit{adjacent} if $u \ne v$ and $u$ and $v$ are connected by an edge. If $u$ and $v$ are adjacent, the length of the edge between them is denoted by $L(u, v)$.

Throughout the paper, we use 0-based indexing. A \textit{path} $p$ is defined to be a sequence of vertices $p_0 \rightarrow p_1 \rightarrow \cdots \rightarrow p_{|p| - 1}$ such that for $0 \le i < |p| - 1$, the vertices $p_i$ and $p_{i + 1}$ are adjacent in $G$. A vertex $u$ is \textit{on} $p$ if $u = p_i$ for $0 \le i < |p|$, and we say $p$ \textit{contains} $u$ or \textit{visits} $u$, and is denoted by $u \in p$.  An edge $e$ is \textit{on} $p$ if the edge is equal to the edge between $p_i$ and $p_{i + 1}$ for $0 \le i < |p| - 1$, and we say $p$ \textit{contains} $e$ or \textit{visits} $e$, and is denoted by $e \in p$. The \textit{length} of $p$ is the sum of the length of the edges on $p$.

The vertices $p_0$ and $p_{|p| - 1}$ are called the \textit{endpoints} of $G$, and we say $p$ is a path \textit{between} $p_0$ and $p_{|p| - 1}$. A vertex on $p$ other than an endpoints is called an \textit{inner vertex} of $p$. For $0 \le i < |p|$, define $\textrm{INDEX}(p, p_i) = i$. A path $p$ is called \textit{simple} if for $0 \le i < j < |p|$, the vertices $p_i$ and $p_j$ are different. The \textit{length} of $p$ denoted by $\|p\|$ is defined by $\|p\| = \sum_{0 \le i < |p| - 1}{L(p_i, p_{i + 1})}$. For $0 \le i \le j < |p|$, $p_{i, j}$ denotes the path $p_i \rightarrow p_{i + 1} \rightarrow \cdots \rightarrow p_j$. Two paths $u$ and $v$ are \textit{different} if $|u| \ne |v|$ or $\exists 0 \le i < |u|, u_i \ne v_i$. Let $p$ be a path. The \textit{reverse} of $p$ is the path $p_{|p| - 1} \rightarrow p_{|p| - 2}  \rightarrow \cdots \rightarrow p_1 \rightarrow p_0\}$. Two paths $u$ and $v$ are called \textit{distinct} if $u$ and both $v$ and the reverse of $v$ are different. For two paths $p_0$ and $p_1$, $p_0$ \textit{contains} $p_1$ if and only if $\exists 0 \le i \le j < |p_0|, p_1 = (p_0)_{i, j}$, and we denote this by $p_1 \subset p_0$. Note that this is different from containment relationship between the sets of edges on these paths.

Let $D$ be a set of edges in $G$. $G \backslash D$ denote the graph we get by removing all edges in $D$ from $G$. If $G$ is connected, a \textit{separating} set of edges is a set of edges $D$ such that $G \backslash D$ is disconnected. A \textit{separating path} is a path such that the set of edges on the path is separating.

We will assume that there is no tie between lengths of paths, and at the end of the paper we will talk about how to break ties in practice. 

\section{Main Algorithm}

Let $G = \langle V, E\rangle$ be a connected chordal graph, and let $S, T$ be distinct vertices in $V$. Due to Theorem \ref{existence}, we can remove all vertices that are un-reachable from $S$ without visiting a bridge without affecting the answer to our problem. It is easy to to see that the graph remains chordal and connected. Therefore, from now on we will suppose the graph does not contain any bridges. 

A simple path $r$ is called a \textit{separator path} if and only if the path is separating and does not properly contain a separating path. For separator paths, we have the following theorem.

\begin{theorem} \label{containsep}
    A simple path $p$ is not separating if and only if the path does not contain a separator path.
\end{theorem}

Since the answer must be a simple path, the answer is not separating if and only if it does not contain a separator path. A separator path $r$ is called \textit{traversable} if there exists a simple path from $S$ to $T$ that contains $r$. Obviously, we do not need to worry about non-traversable separator paths. We will be dealing with traversable separator path of length two separately. A vertex $v$ is called a \textit{bad vertex} if there exists a separator path $r$ with length two such that $r_1 = v$ (i.e. $v$ is the middle vertex). For bad vertices we have the following theorem.

\begin{theorem} \label{badnogood}
    If $u$ is a bad vertex, then no non-separating simple path from $S$ to $T$ visits $u$.
\end{theorem}

With this theorem, if we want to make sure the path does not contain any separator paths of length two. We only need to ensure we don't visit any bad vertices. What remains are separator paths with length more than two. We design a sub-procedure, $\textrm{AVOID}(X)$, where $X$ is a set of separator paths of length more than two, that finds the shortest path from $S$ to $T$ in $G$ that does not contain any element of $X$, \textbf{and does not contain a bad vertex}. Obviously, this shortest path must also be simple. As a shorthand, we say a path \textit{avoids} $X$ if the path does not contain any element of $X$, and does not contain a bad vertex. The sub-procedure runs in $O(T_{\textrm{SPTD}}(2|V|, 8|E|))$ time. There is an important pre-requisite to this sub-procedure: any two separator paths in $X$ must not share a common edge or a common inner vertex.

For separator paths, we have the following theorem:

\begin{theorem} \label{theoadjacent}
    Let $r$ be a separator path. For any $0 \le i < |r| - 2$, the vertex $r_i$ and $r_{i + 2}$ are adjacent. 
\end{theorem}

Let $r$ be a separator path. If for some $0 \le i < |r| - 2$, $\|r_{i, i + 2}\| > L(r_i, r_{i + 2})$, then intuitively, one can simply get around $r$ by going through the edge between $r_i$ and $r_{i + 2}$ instead of going through $r_{i, i + 2}$, and such separator path intuitively should never become an issue. $r$ is called \textit{useful} if for any $0 \le i < |r| - 2$, $\|r_{i, i + 2}\| < L(r_i, r_{i + 2})$. A separator path $r$ is called \textit{normal} if it is both traversable and useful and contains more than two edges. We will formally prove in sub-section \ref{correctnormal}, that any separator path contained in a path that can possibly be produced by $\textrm{AVOID}(X)$ for some set of normal separator paths $X$ is indeed useful, and therefore normal. Normal separator paths have two very important properties and the pre-requisite to the sub-procedure will be met:

\begin{theorem} \label{nosharee}
    Two different normal separator paths do not share a common edge.
\end{theorem}

\begin{theorem} \label{nosharev}
    Two different normal separator paths do not share a common inner vertex.
\end{theorem}

The framework of the algorithm is as follows: we first find a set $X$ of normal separator paths such that one can guarantee that the shortest path that avoids $X$ does not contain a normal separator path not in $X$. Thus $\textrm{AVOID}(X)$ will give us the answer directly. An obvious way is to let $X$ be the set of all normal separator paths in $G$. However, such $X$ will be very hard to compute. A good approach would be to gradually increment $X$: start from $X = \emptyset$ and keep calling $\textrm{AVOID}(X)$ and expanding $X$ to include new unseen separator paths until we find the answer. This gives a simple quadratic time algorithm. However, we will introduce a more efficient approach. We will introduce some special types of normal separator paths that are easy to compute, and we claim that a $X$ that can give us the answer directly can be found doing an extra-sub-procedure based on these separator paths.

Let $P$ be the the shortest path from $S$ to $T$ that does not visit a bad vertex. Let $r$ be a normal separator path $r$ such that for any two vertices $u$ and $v$ shared by $r$ and $P$, $\textrm{INDEX}(r, u) < \textrm{INDEX}(r, v)$ if and only if $\textrm{INDEX}(P, u) < \textrm{INDEX}(P, v)$. $r$ is called an \textit{$S$-separator path} if $r_{2, |r| - 1}$ lies entirely on $P$, and $r$ is called a \textit{$T$-separator path} if $r_{0, |r| - 3}$ lies entirely on $P$. Consider laying $P$ flat with $S$ on one side and $T$ on the other. The edges on an $S$-path seperator can only be off $P$ on the $S$-side, and those on a $T$-separator path can only be off $P$ on the $T$-side. If a separator path is on $P$, it is both an $S$-separator path and a $T$-separator path per definition. 

The framework of the computation for $X$ is as follows: We first the set of bad vertices. We then find the shortest path $P$ from $S$ to $T$ that does not contain any bad vertices. Then, we find the set $X_{ST}$ consisting of the $S$-separator paths and the $T$-separator paths. We will show, in sub-section \ref{badlr}, that the bad vertices, the $S$-separator paths and $T$-separator paths can all be found in $O(|V| + |E|)$ time, by methods based on biconnected components. Now we will compute the set of extra separator paths needed in $X$. Let $r$ be a separator path. The edge between $r_0$ and $r_1$ is called the \textit{head} of $r$, denoted by $\textrm{HEAD}(r)$, and the edge between $r_{|r| - 2}$ and $r_{|r| - 1}$ is called the \textit{tail} of $r$, denoted by $\textrm{TAIL}(r)$. We build a graph $G_0$ which is the maximal subgraph of $G$ that does not contain either of the following:

\begin{itemize}
    \item A bad vertex
    \item The tail of an $S$-separator path
    \item The head of a $T$-separator path
\end{itemize}

Let the shortest path from $S$ to $T$ in $G_0$ be $P_0$. Then $P_0$ corresponds to a simple path in $G$. Let $X_{\textrm{EXTRA}}$ be the set of normal separator paths in $G$ that $P_0$ contains, which as we will show in sub-section \ref{sepcontained} can be computed in $O(|V| + |E|)$ time. We set $X$ to $X_{ST} \cup X_{\textrm{EXTRA}}$. We will prove, in sub-section \ref{correctavoid}, that the set $X$ calculated in this way is such that the shortest path that avoids $X$ does not contain a normal separator path not in $X$.

Here is an example that shows how the algorithm works:

\begin{figure}
    \includegraphics[width=1\textwidth]{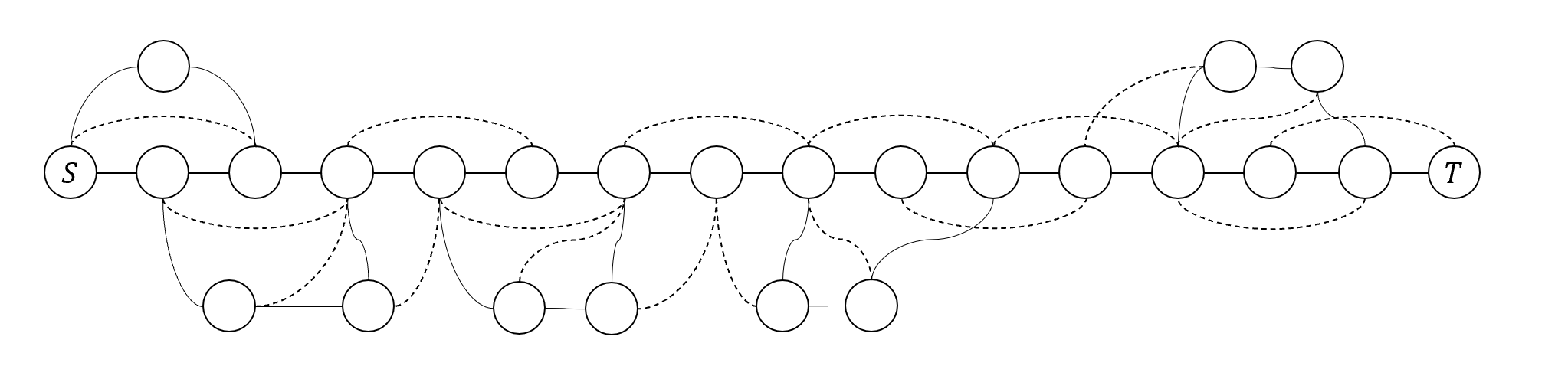}
    \caption{The chordal graph $G$.}
    \label{fig:main0}
\end{figure}

Figure \ref{fig:main1} shows a chordal graph $G$, where dashed edges have length 100 and other edges have length 1. The path $P$ is the horizontal chain in bold from $S$ to $T$ in the middle. 

\begin{figure}
    \includegraphics[width=1\textwidth]{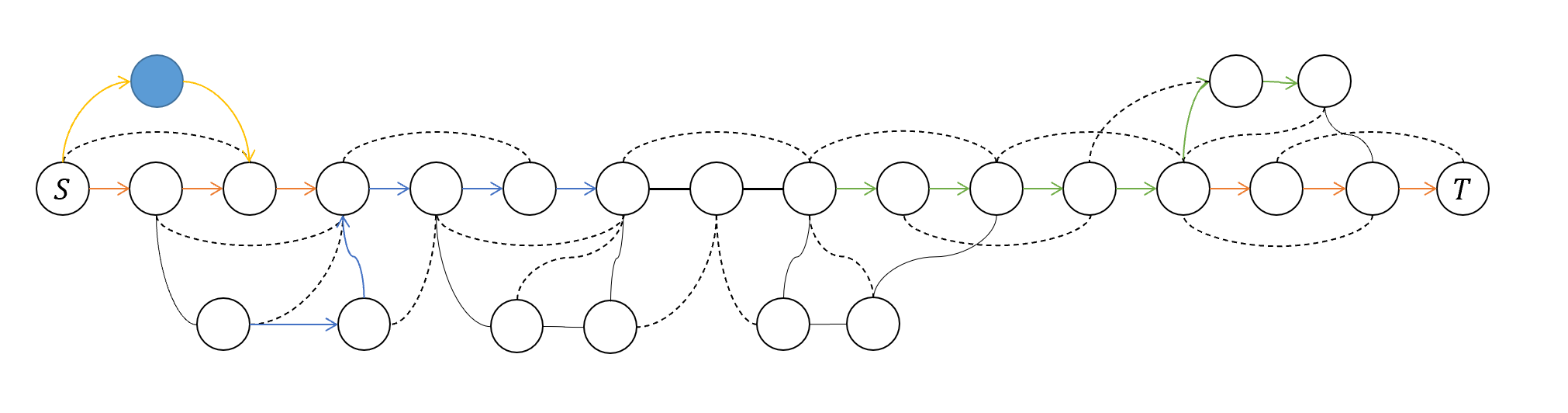}
    \caption{Bad vertices, path $P$ and $X_{ST}$.}
    \label{fig:main1}
\end{figure}

As shown in Figure \ref{fig:main1}, the only bad vertex in the graph is filled in light blue, since it is the middle vertex of the separator path in yellow which is a traversable separator path of length 2. The path $P$ is the horizontal chain from $S$ to $T$ in the middle. The two separator paths in orange are both $S$-separator paths and $T$-separator paths. The separator path in light blue is an $S$-separator path and the separator path in green is a $T$-separator path. $X_{ST}$ contains the set of $S$-separator paths and $T$-separator paths (i.e. all separator paths shown except the one in yellow).

\begin{figure}
    \includegraphics[width=1\textwidth]{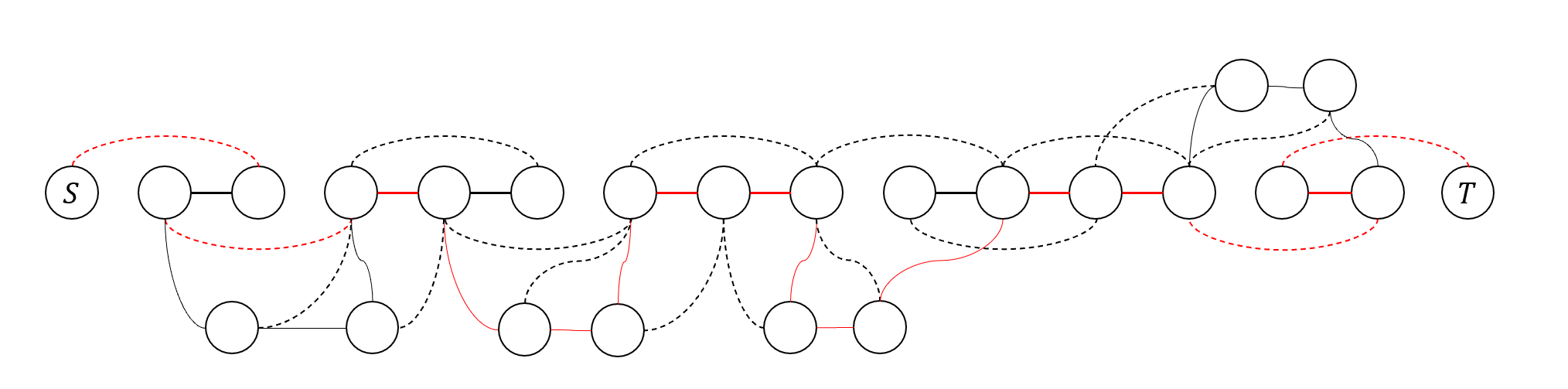}
    \caption{The graph $G_0$ and the path $P_0$.}
    \label{fig:main2}
\end{figure}

Figure \ref{fig:main2} shows the graph $G_0$, the maximal subgraph of $G$ that does not contain a bad vertex, the tail of an $S$-separator path or the head of a $T$-separator path, and $P_0$, the shortest path from $S$ to $T$ in $G_0$ in red.

\begin{figure}
    \includegraphics[width=1\textwidth]{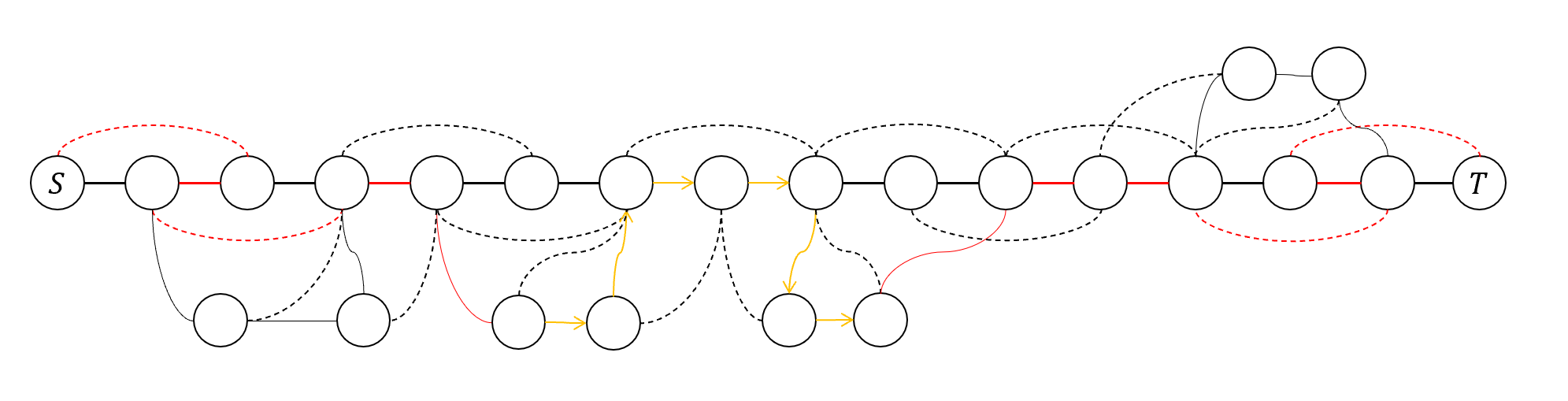}
    \caption{$P_0$ shown in $G$ and $X_{\textrm{EXTRA}}$.}
    \label{fig:main3}
\end{figure}

Figure \ref{fig:main3} shows the path $P_0$ in $G$ in red, and it contains an extra normal separator path in yellow.

\begin{figure}
    \includegraphics[width=1\textwidth]{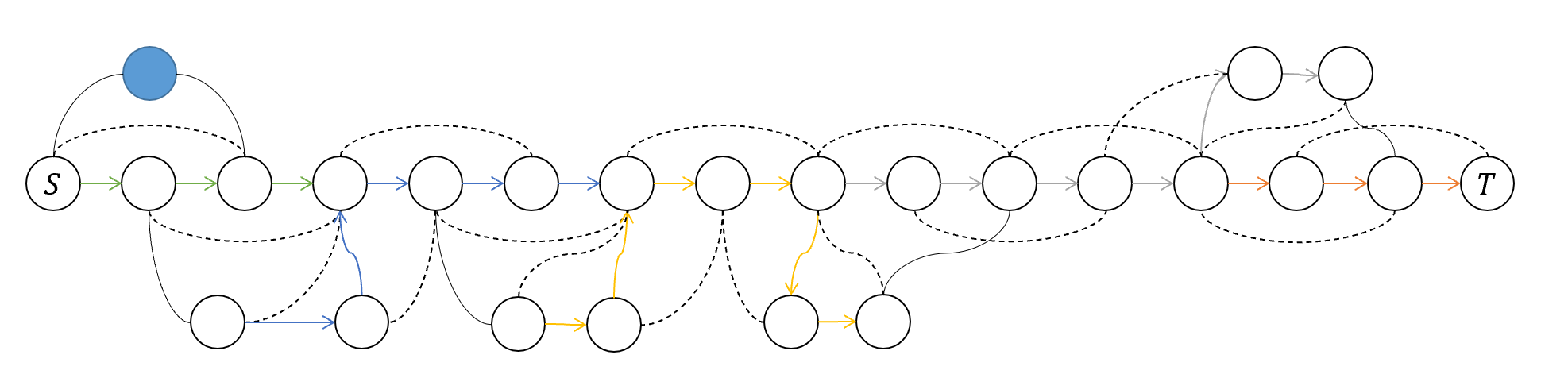}
    \caption{The graph $G$, bad vertices and the set $X$.}
    \label{fig:main4}
\end{figure}

Figure \ref{fig:main4} shows graph $G$, the bad vertex in light blue and the set $X$, the only separator path in $X_{\textrm{EXTRA}}$ is the one in yellow and the other separator paths are in $X_{ST}$. The algorithm then proceeds to calculate $\textrm{AVOID}(X)$ which will give us the answer.

The pseudocode is given below, where $E(P)$ denotes the set of edges on $P$: 
 
\begin{algorithm}
    \caption{Main Algorithm}
    \begin{algorithmic}[1]
        \Require Chordal Graph $G = \langle V, E\rangle$, $S, T \in V (S \ne T)$
        \Ensure Shortest Non-Separating Path Between $S$ and $T$
        \State Calculate $\textrm{BAD}_G$, the set of bad vertices in $G$
        \State $G ^ {\prime} = \langle V ^ {\prime}, E ^ {\prime}\rangle \gets$ the subgraph of $G$ induced by $V \backslash \textrm{BAD}_G$
        \State $P \gets$ the shortest path from $S$ to $T$ in $G ^ {\prime}$
        \State $X_{ST} \gets$ the set of $S$-separator paths and $T$-separator paths
        \State Let $G_0$ be the induced subgraph of $E ^ {\prime} \backslash (\{\textrm{TAIL}(r) \mid r\textrm{ is an }S\textrm{-separator path}\}) \backslash (\{\textrm{HEAD}(r) \mid r\textrm{ is an }T\textrm{-separator path}\})$ of $G ^ {\prime}$.
        \State $P_0 \gets$ the shortest path from $S$ to $T$ in $G_0$
        \State $X_{\textrm{EXTRA}} \gets $ the set of normal separator paths in $G$ that $P_0$ contains
        \State $X \gets X_{ST} \cup X_{\textrm{EXTRA}}$ \\
        \Return $\textrm{AVOID}(X)$.
    \end{algorithmic}
\end{algorithm}

\newpage

\section{Properties of Separator Paths and Proof}

\subsection{An Important Lemma and an Important Theorem}

We have already seen Theorem \ref{theoadjacent}, which states that for a separator path $r$, for any $0 \le i < |r| - 2$, the vertex $r_i$ and $r_{i + 2}$ are adjacent. In fact, more properties of separator paths exist.

Let $p$ be a simple path. Two vertices $u$ and $v$ are called weakly $p$-connected if there exists a path between $u$ and $v$ that does not share a common edge with $p$. $u$ and $v$ are called strongly $p$-connected if there exists a path $p ^ {\prime}$ between $u$ and $v$ such that except at the endpoints, $p ^ {\prime}$ does not share a common vertex with $p$, and $p ^ {\prime}$ does not share a common edge with $p$. We have the following important lemma and theorem:

\begin{lemma} \label{lemmapathseparator1}
    Let $r$ be a separator path. For all $0 \le i < |p| - 2$, let $j$ be the minimum index such that $i < j < |p|$ and $r_i$ and $r_j$ are weakly $r$-connected, then $j = i + 2$. For $i \ge |p| - 2$, no such index exists.
\end{lemma}

\begin{theorem} \label{pathseparatorprop1}
    Let $r$ be a separator path. For all $0 \le i < j < |r|$, $r_i$ and $r_j$ are strongly $r$-connected if and only if $j - i = 2$.
\end{theorem}

\begin{figure}
    \includegraphics[width=1\textwidth]{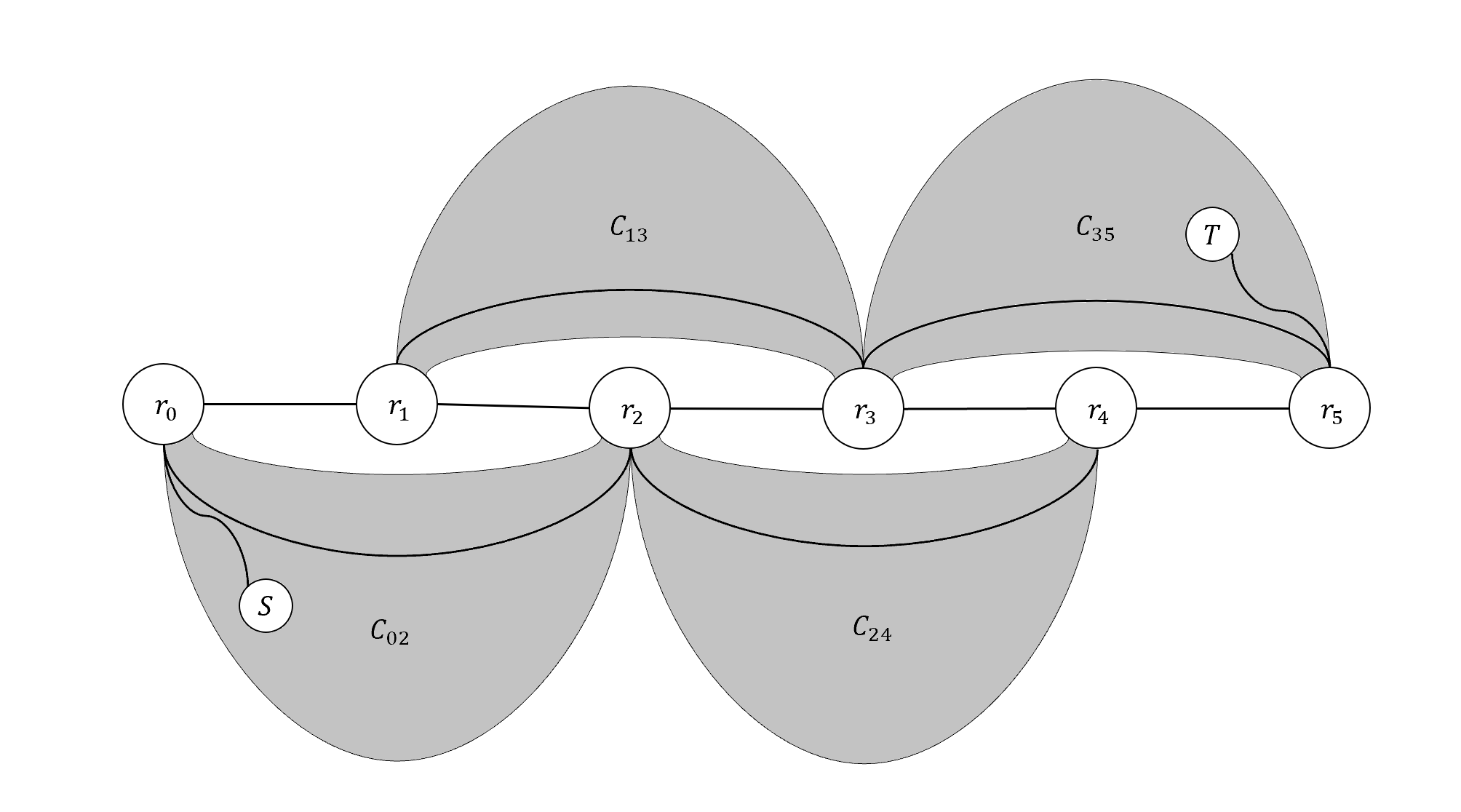}
    \caption{Separator path visualization.}
    \label{fig:separator}
\end{figure}

To help understand these properties of separator paths, we have a visualization for a separator path $r$ with 6 vertices in Figure \ref{fig:separator}. Recall that due to Theorem \ref{theoadjacent}, There is an edge between $r_i$ and $r_{i + 2}$ for $0 \le i < |r| - 2$. we can see that the relation of $r$ to the graph is as follows. There are five components $C_{02}, C_{13}, C_{24}, C_{35}$. For component $C_{uv}$, all vertices in $C_{uv}$ can not be strongly $r$-connected to any vertices on $r$ except $r_u$ and $r_v$, and. If a vertex is only strongly $r$-connected to a single $r_x$ on $r$, add the vertex to an arbitrary component involving $r_x$. For example if $y$ is strongly $r$-connected to $r_2$ only, $y$ can be in either $C_{02}$ or $C_{24}$. $C_{uv}$ also include edges from vertices inside it to $r_u$ and $r_v$, and the edge between $r_u$ and $r_v$. Moreover, it is not hard to see that if $r$ is traversable, $S \in C_{02}$ and $T \in C_{35}$. Moreover, there exists a path from $S$ to $r_0$ which is entirely inside $C_{02}$ and a path from $r_5$ to $T$ which is entirely inside $C_{35}$. Note that $S$ may or may not be strongly $r$-connected to $r_2$ and the same goes for $T$ and $r_3$. For any separator path in a connected chordal graph, one can draw a diagram like this.

Theorem \ref{lemmapathseparator1} has many corollaries.

\begin{cor} \label{corollary:innervisit}
    If $r$ is a normal separator path, and $p$ is a path from $S$ to $T$. At least one inner vertex of $r$ is on $p$
\end{cor}

\begin{proof}
	if $p$ does not visit an inner of $r$, then $r_0$ and $r_{|r| - 1}$ are strongly $r$-connected. Since $|r| > 3$, this violates Theorem \ref{pathseparatorprop1}.
\end{proof}

\begin{cor}\label{noreverse}
    Let $r$ be a normal separator path, then $r_0$ and $T$ are not strongly $r$-connected. $r_{|r| - 1}$ and $S$ are not strongly $r$-connected. Consequently $S$ and $T$ are not strongly $r$-connected. Consequently the reverse of $r$ is not a normal separator path.
\end{cor} 

\begin{proof}
    We know $r_{|r| - 1}$ and $T$ are strongly $r$-connected and $|r| > 3$, if $r_0$ and $T$ are strongly $r$-connected then $r_0$ and $r_{|r| - 1}$ are strongly $T$ connected, which violates Theorem \ref{pathseparatorprop1}.
    
    By the same reasoning, $r_{|r| - 1}$ and $S$ are not strongly $r$-connected.
\end{proof}

\begin{cor} \label{corollary:normalvisit}
    If $r$ is a normal separator path, for all $2 \le i < |r| - 1$, all paths from $S$ to $T$ must visit at least one of $r_{i - 1}$ and $r_i$.
\end{cor}

\begin{proof}
	Let $p$ be a path from $S$ to $T$. From Corollary \ref{corollary:innervisit} we know $p$ must visit an inner vertex on $r$. Suppose neither $r_{i - 1}$ nor $r_i$ is visited. If no inner vertex that $p$ visits precedes $r_{i - 1}$ on $r$, then the first inner vertex that $p$ visits is $r_k$ for $k > 2$, but $r_k$ and $r_0$ will be strongly $r$-connected and this violates Theorem \ref{pathseparatorprop1}. Similarly $p$ must visit an inner vertex after $r_i$. Find the last $r_u (u < i - 1)$ that $p$ visits and the first $r_v (v > i)$ that $p$ visits, and $r_u$ and $r_v$ are strongly $r$-connected. Since $v - u > 2$, this violates Theorem \ref{pathseparatorprop1}.
\end{proof}

\begin{cor} \label{corollary:firstvisit}
    If $r$ is a normal separator path, and $p$ is a path from $S$ to $T$. The first vertex on $r$ that $p$ visits is either $r_0$ or $r_2$, and the last vertex on $r$ that $p$ visits is either $r_{|r| - 3}$ or $r_{|r| - 1}$.
\end{cor}

\begin{proof}
	If the first vertex on $r$ that $p$ visits is $r_k$ where $k \ne 0$, then $r_k$ and $r_0$ are strongly $r$-connected. From Theorem \ref{pathseparatorprop1}, this is only possible if $k = 2$. The other direction is similar.
\end{proof}

\begin{cor} \label{traorder}
    Let $p$ be a simple path from $S$ to $T$ and let $r$ be a normal separator path. For $0 < i < |r| - 1 (i \ne 2)$, if $p_j = r_i (j < |p| - 1)$ and $p_{j + 1} \notin \{r_{i - 1}, r_{i + 1}\}$, then for all $k > j$ and $l \le i$, $p_k \ne r_l$.
\end{cor} 

\begin{proof}
    If $p_{j + 1} \notin \{r_{i - 1}, r_{i + 1}\}$, then the next visit of a vertex on $r$ after $r_i$ on $p$, if it exists, must be either $r_{i - 2}$ or $r_{i + 2}$. In the former case, $p$ hasn't visited $r_{i - 2}$ before visiting $r_i$. We have $i \ge 2$ and $i \ne 2$ and therefore $i > 2$. Since $r_{i - 2}$ is not visited before $r_i$, one can see that $p$ must visit $r_{i - 1}$ before $r_i$. Now that we are currently at $r_{i - 2}$ and both $r_i$ and $r_{i - 1}$ are already visited, it will be impossible for us to reach $T$ again while the path remains simple. In the latter case, we have $i < |r| - 2$, and one can see that with $r_i$ already visited, in order to visit some $r_l (l < i)$ again, we need to visit $r_{i + 1}$ first. Now that we are at $r_l$, with $r_i$ and $r_{i + 1}$ both visited and $i < |r| - 2$, it will be impossible for us to reach $T$ again while the path remains simple.
\end{proof}

\subsection{Proof of Key Theorems}

First of all, we want to prove Theorem \ref{containsep}, which is dependent on the following lemma.

\begin{lemma} \label{minsepconnected}
    If $D$ is a minimal separating set of edges, then the induced subgraph of $D$ is connected.
\end{lemma}

\begin{proof}[Proof of Lemma \ref{minsepconnected}]
    If $G \backslash D$ contains more than two connected components, we can take two connected components and the edges between them in $G$ will be a separating proper subset of $D$, which is a contradiction. Therefore $G \backslash D$ contains exactly two connected components.
    
    Let the two connected components of $G \backslash D$ be $C_A$ and $C_B$. If $D$ is disconnected, then we can find $u, x \in C_A (u \ne x)$ and $v, y \in C_B (v \ne y)$ such that $u$ and $v$ are adjacent in $G$, and so are $v$ and $y$.  Let $p_{ux}$ be the path between $u$ and $x$ with the minimum amount of edges, and $p_{vy}$ be such path between $v$ and $y$. The cycle formed by $p_{ux}$, $p_{vy}$, the edge between $u$ and $v$ and the edge between $x$ and $y$ will form a simple cycle of length at least four without a chord, which contradicts the fact that $G$ is chordal.
\end{proof}

\begin{proof}[Proof of Theorem \ref{containsep}]
    If $p$ contains a separator path then $p$ is obviously separating.

    If $p$ is separating, let $D$ be a minimally separating subset of the edges in $p$. If $D$ corresponds to $p_{i, j}$ for some $0 \le i < j < |p|$, then $p_{i, j}$ is a separator path. Otherwise $D$ is disconnected, which is impossible due to Lemma \ref{minsepconnected}.
\end{proof}

\begin{proof} [Proof of Theorem \ref{badnogood}]
    If $r$ is a traversable separator of length 2 and a path $p$ between $S$ and $T$ visits $r_1$, we can show that $p$ must contain the edge between $r_0$ and $r_1$. Suppose it does not. If $r_1$ is the first or the last vertex on $r$ that $p$ visits, then $r_1$ and either $S$ or $T$ are strongly $r$-connected, and therefore $r_1$ and either $r_0$ or $r_2$ are strongly $r$-connected, violating Theorem \ref{pathseparatorprop1}. Otherwise, $r_1$ must be either the last vertex on $r$ that $p$ visits before $r_1$ or the first vertex on $r$ that $p$ visits after $r_1$, and $r_0$ and $r_2$ are strongly $r$-connected, violating Theorem \ref{pathseparatorprop1}. By the same reasoning $p$ must contain the edge between $r_1$ and $r_2$. Therefore both edges on $r$ will be contained in the path $p$ and the path will not be non-separating.
\end{proof}

\begin{proof}[Proof of Lemma \ref{lemmapathseparator1}]

    If $i = |p| - 1$ then no index greater than $i$ exists.
    
    If $i = |p| - 2$ and $r_i$ and $r_{i + 1}$ are weakly $r$-connected, then since $r$ is separating, $r_{0, |r| - 2}$ is separating. Therefore $r$ isn't minimal, which is a contradiction. Similarly, we can conclude that $r_0$ and $r_1$ are not weakly $r$-connected.
    
    If $j - i > 2$ and $r_i$ and $r_j$ are weakly $r$ connected but for all $i < k < j$, $r_i$ and $r_k$ are not weakly $r$-connected, $r_i$ and $r_j$ are strongly $r_{i, j}$-connected. Consider the simple cycle consisting of $r_{i, j}$ and the path $p$ with the minimum amount of edges between $r_i$ and $r_j$ that do not share a common vertex with $r_{i, j}$ except for $p$'s endpoints. This cycle is simple and has a length greater than three and therefore must have a chord. By the way the path $p$ is chosen, the chord can't have both endpoints on $p$, and therefore an endpoint of the chord must be $r_k$ for some $i < k < j$. This makes $r_i$ and $r_k$ weakly $r$-connected, which is a contradiction.
    
    If $r_i$ and $r_{i + 1}$ are weakly $r$-connected for some $1 \le i < |p| - 1$, find the maximum $a < i$ such that $r_a$ and $r_i$ are not weakly $r$-connected and the minimum $b > i$ such that $r_b$ and $r_i$ are not weakly $r$-connected. Since $r_0$ and $r_1$, $r_{|r| - 2}$ and $r_{|r| - 1}$ are not weakly $r$-connected, both $a$ and $b$ exist. Let $R$ be the set of edges on $r$. From the fact that $R$ is a minimum separating set of edges, and from Lemma \ref{minsepconnected} we know $G \backslash R$ contains exactly two connected components. Therefore $r_a$ and $r_b$ are weakly $r$-connected, but since $b - a > 2$, this has been shown to be impossible.
\end{proof}

\begin{proof} [Proof of Theorem \ref{theoadjacent}]
	For $0 \le i < |r| - 2$, from Lemma \ref{lemmapathseparator1}, $r_i$ and $r_{i + 2}$ are strongly $r_{i, i + 2}$-connected. If $r_i$ and $r_{i + 2}$ are not adjacent, consider the simple cycle consisting of $r_{i, i + 2}$ and the path $p$ with the minimum amount of edges between $r_i$ and $r_{i + 2}$ that don't share a common vertex with $r_{i, i + 2}$ except the endpoints. The cycle will have at least four edges. Using a similar reasoning as the one in the proof of Lemma \ref{lemmapathseparator1}, an endpoint of the chord is $r_{i + 1}$. It is easy to see that the other endpoint will be on $p$ and the chord makes $r_{i + 1}$ and $r_i$ weakly $r$-connected, which is impossible from Lemma \ref{lemmapathseparator1}.
\end{proof}

\begin{proof} [Proof of Theorem \ref{pathseparatorprop1}]
	
	For $0 \le i < j < |r|$, if $j - i = 2$, then from Theorem \ref{theoadjacent}, $r_i$ and $r_j$ are strongly $r$-connected. 

	Reversely, if $r_i$ and $r_j$ are strongly $r$-connected and therefore weakly $r$-connected, from Lemma \ref{lemmapathseparator1} we can easy see that $j - i$ is even. Let $i$ and $j$ be such that for there exists no $i \le k < l \le j$ such that $(i, j) \ne (k, l)$ and $r_k$ and $r_l$ are strongly $r$-connected.
    
    If $j - i > 4$, consider the cycle consisting of the edge between $r_i$ and $r_{i + 2}$, the edge between $r_{i + 2}$ and $r_{i + 4}$, $\cdots$, the edge between $r_{j - 2}$ and $r_j$, and the path $p$ with the minimum amount of edges between $r_i$ and $r_j$ that don't share a common vertex with $r$ except for $p$'s endpoints. The cycle contains more than three edges and it is straightforward, from the similar reasoning we had in the proof of Lemma \ref{lemmapathseparator1}, that any chord of the cycle will imply the existence of $i \le k < l \le j$ such that $(i, j) \ne (k, l)$ and $r_k$ and $r_l$ are strongly $r$-connected.
    
    If $j - i = 4$, consider the cycle consisting of the edge between $r_i$ and $r_{i + 1}$, the edge between $r_{i + 1}$ and $r_{i + 3}$, the edge between $r_{i + 3}$ and $r_{i + 4}$, and the path $p$ with the minimum amount of edges between $r_i$ and $r_j$ that don't share a common vertex with $r$ except for $p$'s endpoints. The cycle has more than three edges and we can see that any chord of the cycle will imply that either $r_i$ and $r_{i + 3}$ are strongly $r$-connected or $r_{i + 1}$ and $r_{i + 4}$ are strongly $r$-connected, neither of which is impossible since both pairs of indices have different parities.
\end{proof}

\begin{proof}[Proof of Theorem \ref{nosharee}]

    Due to Corollary \ref{noreverse} it suffices to show that two distinct (disregarding direction) normal separators don't share a common edge. In fact, the theorem applies to two distinct useful separators.

    Let $u$ and $v$ be distinct useful separator paths and $u$ and $v$ share a common edge. Let $0 \le i < j < |u|$ be such that $u_{i, j}$ is entirely on $v$ and $j - i$ is maximal. Since the set of edges $u$ is not contained in that of $v$, either $i > 0$ or $j < |p| - 1$. Without loss of generality, assume $i > 0$ and that $\textrm{INDEX}(v, u_{i + 1}) = \textrm{INDEX}(v, u_i) + 1$. Let $x = \textrm{INDEX}(v, u_i)$.
    
    If $x = 0$, then $v_0 = u_i$ and $v_1 = u_{i + 1}$. From Theorem \ref{theoadjacent} we know $u_{i - 1}$ and $u_{i + 1}$ are adjacent. If the edge is on $v$, then we violate Lemma \ref{lemmanoshare}; otherwise since the edge between $u_{i - 1}$ and $u_i$ is not on $v$, $v_0$ and $v_1$ are strongly $v$-connected, which is a contradiction of Theorem \ref{pathseparatorprop1}.
    
    If $x > 0$, then $v_x = u_i$ and $v_{x + 1} = u_{i + 1}$. If $v_{x - 1} \notin u$, then $v_{x - 1}$ and $v_{x + 1}$, $v_{x - 1}$ and $v_x$ are both connected by an edge not on $u$, and therefore $u_i$ and $u_{i + 1}$ are strongly $u$-connected, which is a contradiction of Theorem \ref{pathseparatorprop1}.
\end{proof}

\begin{proof} [Proof of Theorem \ref{nosharev}]

    Due to Corollary \ref{noreverse} it suffices to show this for two distinct normal separators. Let $u$ and $v$ be two distinct normal separators.

    Let $u_i (1 \le i < |u| - 1)$ and $v_j (1 \le j < |v| - 1)$ be the same vertex. We argue that for some $u_a$ and $u_b$ such that $a \le i - 2$ and $b \ge i + 2$, $u_a$ and $v_{j - 1}$ are strongly $u$-connected and $u_b$ and $v_{j + 1}$ are strongly $u$-connected. We know $v_{j - 1}$ and $v_{j + 1}$ are adjacent, and since $u$ and $v$ don't share a common edge due to Theorem \ref{nosharee}, neither of $v_{j - 1}$ and $v_{j + 1}$ is on $u$ (otherwise one can see that $v_{j - 1}$ or $v_{j + 1}$ and $v_j$ are weakly $v$-connected, violating Lemma \ref{lemmapathseparator1}. Therefore, since $u_a$ and $v_{j - 1}$ are strongly $u$-connected and $u_b$ and $v_{j + 1}$ are strongly $u$-connected, $u_a$ and $u_b$ are strongly $u$-connected, but since $b - a \ge 4$, this violates Theorem \ref{pathseparatorprop1}.

	By symmetry we will only show the existence of $a$. Let $p$ be a path from $S$ to $T$ that contains $v$. If the part of $p$ before $v_j$ does not contain a vertex on $u$, then $S$ and $u_0$ are strongly $u$-connected and $S$ and $v_{j - 1}$ are strongly $u$-connected, and $u_0$ and $v_{j - 1}$ are strongly $u$-connected. Therefore, we can let $a = 0$. If $i = 1$, that would imply $u_0$ and $u_1$ are strongly $u$-connected, which is impossible from Theorem \ref{pathseparatorprop1}. Therefore $i > 1$ and $a \le i - 2$.
    
    If the last vertex on $u$ that $p$ visits before $v_j$ is $u_j$, since from Theorem \ref{nosharee} $u$ and $v$ don't share a common edge, $u_j$ and $u_i$ are strongly $u$-connected. If $j = 0$, we can use the same reasoning above to deduce $i > 1$ and we can set $a = 0$. If $j > 0$, from Corollary \ref{traorder} we know $i - j$ is -1, or, 1 or $\ge 2$. Since $u_j$ and $u_i$ are strongly $u$-connected, from Theorem \ref{pathseparatorprop1} we know $|j - i| \ne 1$ and therefore $j = i - 2$. Therefore we can set $a = i - 2$.
    
\end{proof}

\section{Proof of Theorem \ref{existence}} \label{existenceproof}

This short section completes the proof of Theorem \ref{existence}, our result pertaining to the decision problem.

It suffices to show that the separator paths on the path that visits the minimum amount of edges have length one (i.e. bridges).
    
Suppose there is a separator path on this path with length more than one, then Theorem \ref{theoadjacent} shows that two non-adjacent vertices on the path are adjacent, this contradicts the fact that the path visits the minimum amount of edges. Therefore, all separator paths on the path that visits the minimum amount of edge are bridges.

\section{Sub-procedures}

In this section, we will detail on the sub-procedures required in the main algorithm.

\subsection{Shortest Path That Avoids a Set of Normal Separator Paths} \label{avoid}

\subsubsection{The Sub-procedure} 

The sub-procedure $\textrm{AVOID(X)}$ takes a set $X$ of normal separator paths in the chordal graph $G$, and computes the shortest path from $S$ to $T$ in $G$ that avoids $X$. In order to do this, we build an auxiliary \textbf{directed} graph $G_1(X) = \langle V_1, E_1\rangle$ that contains no more than two copies of each vertex in $V$. Specifically, both $S$ and $T$ only have one copy. Of course, we want to ensure that we don't visit any bad vertices. Therefore, $G_1$ does not contain any copies of bad vertices. It also contains no more than four copies of each edge in $G$, for either direction, yielding a total of no more than $2|V|$ vertices and $8|E|$ edges. The shortest path from the vertex corresponding to $S$ in $G_1$ to the vertex corresponding to $T$ in $G_1$ will correspond to a shortest path from $S$ to $T$ in $G$ that avoids $X$.

The motivation behind $G_1$ is as follows: suppose we are halfway on a path from $S$ to $T$ and we are at a vertex that is an inner vertex $r_i$ of some normal separator path $r$. There are two possibilities: the part of the path before might already contain the entirety of $r_{0, i}$, which means that we do not want to cover the entirety of $r_{i, |r| - 1}$, and we need to get off $r$. Or the path before might not contain the entirety of $r_{0, i}$. In this case if we assume our path is simple and will remain so, it would mean that we are ``safe'' and do not need to worry about containing $r$. Therefore for each inner vertex $u$ of a normal separator path, the graph $G_1$ contains a \textit{high vertex} $(u, 1)$, which indicates the first, unsafe case, and a \textit{low vertex} $(u, 0)$, which indicates the second, safe case. Neither $S$ nor $T$ is an inner vertex of a normal separator path, so only low vertices of them exist. 

In order to make this work, for any separator path $r$, the edge from $(r_0, 0)$ to $(r_1, 0)$ must not exist, so that we can't get on $r$ without entering the high vertex. The edge between $(r_{|r| - 2}, 1)$ and $(r_{|r| - 1}, 1)$ must not exist, since going through this edge would imply containing the entirety of $r$. Moreover, we are forbidden to go from $(r_i, 1)$ to $(r_{i + 1}, 0)$ so that we can't get off the high vertices without getting off $r$. To help simplify the formal proof in sub-section \ref{correctavoid}, we forbid the edges to the corresponding high vertices as well.

However, if we simply keep all the other edges. we still have a problem: since we now have two copies of a vertex. A simple path in the new graph does not necessarily correspond to a simple path in the original graph. For an inner vertex $r_i$ on separator path $r$, one can go from the high vertex $(r_i, 1)$ to $(x, 0)$ for some $x \notin r$, and then go to the low vertex $(r_i, 0)$ from $(x, 0)$. We will introduce a fix. For $v \notin p$. Define $L(p, v)$ to be the $p_i$ with the smallest index $i$ such that $p_i$ is strongly $p$-connected to $v$, and define $R(p, v)$ to be the $p_i$ with the largest index $i$ such that $p_i$ is strongly $p$-connected to $v$. From Theorem \ref{pathseparatorprop1}, if $r$ is a separator path, then for any $r_i$ and $v \notin r$, at least one of $r_i = L(r, v)$ and $r_i = R(r, v)$ is true. In section \ref{lr} we will show that the values of $L(r, u)$ and $R(r, v)$ required in the procedure can be efficiently computed. For the moment, assume that these values are accessible. Consider Corollary \ref{traorder}, for most $r_i$, if we go from an inner vertex $r_i$ to a vertex $x \notin r$ and want to keep our path (in the original graph) simple, then we can not ever visit a vertex on $r_{0, i}$. If $R(r, x) = r_i$ and we are at the high vertex for $r_i$, then we can not go to either copy of $x$ and keep our path simple, and therefore we forbid both edges. Similarly, if $L(r, x) = r_i$, suppose we are at a copy of vertex $x$, then we must have either visited $r_i$ before, or one can verify that we have visited $r_{i + 1}$ and $r_{i + 2}$ before. In the former case the path will not be simple, and in the latter case the path will not be able to reach $T$ while the path remains simple. Therefore we forbid the edge from $(x, 0)$ or $(x, 1)$ (if it exists) to $(r_i, 0)$. This will forbid the scenario where one leaves the high vertex $(r_i, 1)$ and goes around a cycle that does not contain another vertex on $r$ back to $(r_i, 0)$. Recall that for each vertex $u \ne r$, at least one of $L(r, u) = r_i$ and $R(r, u) = r_i$ is true. Therefore either the last edge on the cycle back to the low vertex $r_i$ is forbidden, or the first edge on the cycle leaving the high vertex $r_i$ is forbidden. To prevent similar exploits, for $i > 0$, we are also forbidden to go from $(r_{i + 1}, 1)$ to $(r_i, 0)$ or $(r_{i - 1}, 0)$. 

There is are two caveats. Firstly, it is possible to leave $r_i$ and then immediately visit the head of another separator path. In this case, if we are going from $r_i$ to $u = r ^ {\prime}_1$ for another separator path $r ^ {\prime}$, and $R(r, u) \ne r_i$, we go from $(r_i, 1)$ to $(u, 1)$ instead of $(u, 0)$. Secondly if $i = |r| - 3$ and $L(r, x) = r_i$, we might have gotten off $r$ from $r_{|r| - 1}$, which is not an inner vertex and Corollary \ref{traorder} no longer applies. If this is the case, since we can not visit $r_{|r| - 1}$ again, we will have to get off $r$ again, and therefore we allow the edge from $(x, 0)$ to the high vertex $(r_i, 1)$. One can show that even with this edge, whenever we are at $(r_i, 1)$, all vertices on $r_{0, i}$ have been visited.  A similar issue exists for $i = 2$. Fortunately if we get off $r$ from $r_2$ and visits a vertex on $r_{0, 2}$, the first vertex we visit will be $r_0$. This implies that we must currently be on the low vertex of $r_2$ since otherwise $r_0$ is already visited. Our graph does not forbid going (directly or indirectly) from the low vertex of $r_2$ to $r_0$ and therefore our auxiliary graph still admits such paths. Currently, one might have concerns about the correctness of this sub-procedure, especially because we can have multiple separator paths. Fortunately, Theorem \ref{nosharee} and Theorem \ref{nosharev} state that these separator paths do not share an edge of an inner vertex, and a formal proof of the sub-procedure will be given in sub-section \ref{correctavoid} which makes use of these two theorems.

The pseudocode is given below. For a separator path $r$, the set of inner vertices is denoted by $\textrm{INNER}(r)$. 

\begin{algorithm}
    \caption{Computation of $\textrm{AVOID(X)}$, The Shortest Path Avoiding $X$}
    \begin{algorithmic}[1]
        \Procedure{AVOID}{$X$}
            \State $V_1 \gets ((V \backslash \textrm{BAD}_G) \times \{0\}) \cup ((\cup_{r \in X}{\textrm{INNER}(r)} \backslash \textrm{BAD}_G) \times \{1\})$
            \State $E_1 \gets \emptyset$
            \For {ordered pair $(u, v) \in V$ where $u$ and $v$ are connected by an edge $e$ of length $w$ in $G$}
                \If {$u \notin \textrm{BAD}_G$ \AND $v \notin \textrm{BAD}_G$}
                    \If {$\exists r \in X$, (($u \notin r$ \AND $v = r_{|r| - 3}$) \OR ($u = r_0$ \AND $v = r_1$))}
                        \State Add an edge from $(u, 0)$ to $(v, 1)$ with length $w$ to $E_1$
                    \EndIf
                    \If {\NOT ($\exists r \in X, (u, v) = (r_0, r_1)$ \OR ($\exists r \in X$, $u \notin r$ \AND $v \in \textrm{INNER}(r)$ \AND $L(r, u) = v$))}
                        \State Add an edge from $(u, 0)$ to $(v, 0)$ with length $w$ to $E_1$
                    \EndIf
                    \If {$(u, 1) \in V_1$ \AND $(v, 1) \in V_1$ \AND \NOT (($\exists r \in X, (u, v) = (r_{|r| - 2}, r_{|r| - 1})$) \OR ($\exists r \in X, 2 \le i < |r| - 1, (u, v) = (r_i, r_{i - 2})$ \OR $(u, v) = (r_i, r_{i - 1})$) \OR ($\exists r \in X$, $v \notin r$ \AND $u \in \textrm{INNER}(r)$ \AND $R(r, v) = u$))}
                        \State Add an edge from $(u, 1)$ to $(v, 1)$ with length $w$ to $E_1$
                    \EndIf
                    \If {$(u, 1) \in V_1$ \AND \NOT (($\exists r \in X, e \in r$) \OR ($\exists r \in X, 2 \le i < |r| - 1, (u, v) = (r_i, r_{i - 2})$) \OR ($\exists r \in X$, (($u \notin r$ \AND $v \in \textrm{INNER}(r)$ \AND $L(r, u) = v$) \OR ($v \notin r$ \AND $u \in \textrm{INNER}(r)$ \AND $R(r, v) = u$)))}
                        \State Add an edge from $(u, 1)$ to $(v, 0)$ with length $w$ to $E_1$
                    \EndIf
                \EndIf
            \EndFor 
            \State $p \gets$ the shortest path from $(S, 0)$ to $(T, 0)$ in $G_1(X) = \langle V_1, E_1\rangle$
            \For {$0 \le i < |p|$}
                \State Replace $p_i$ with the vertex in $G$ that $p_i$ corresponds to
            \EndFor \\
            \Return $p$
        \EndProcedure
    \end{algorithmic}
\end{algorithm}

\newpage

\subsubsection{Calculation of $L(r, u)$ and $R(r, v)$} \label{lr}

We show a way to calculate $L(r, u)$ for all queries involved in the algorithm. The queries for $R(r, v)$ can be calculated similarly.

We notice that every time we want to calculate $L(r, u)$, there must be a vertex $v \in \textrm{INNER}(r)$ that is adjacent to $u$. Let $v = r_i$. From Theorem \ref{pathseparatorprop1}, the answer can only be $r_i$ or $r_{i - 2}$. All that we need to do is to check if $r_{i - 2}$ is strongly $r$-connected to $u$. Our problem becomes, given $r_i$, a vertex $u$ adjacent to $r_i$, decide whether $u$ is strongly $r$-connected to $r_{i - 2}$.

Firstly, if $|X| = 1$ and the only element in $X$ is the separator path $r$. Then let $G ^ {\prime}$ be the graph after we removing all edges on $r$ from $G$. One can see, from Theorem \ref{pathseparatorprop1} and Theorem \ref{theoadjacent}, that vertices $r_k$ and $r_l$ share a biconnected component if and only if $|k - l| = 2$. From this we can see that given there is an edge between $u$ and $r_i$, $u$ is strongly $r$-connected to $r_{i - 2}$ if and only if $u$ share a biconnected component with both $r_i$ and $r_{i - 2}$. This can be done by checking if the biconnected component containing the edge between $r_i$ and $u$ contains $r_{i - 2}$.

For the case where $|X| > 1$, calling this sub-procedure for all $r \in X$ will be too costly. We want to calculate the values for all $r \in X$ at once. We let $G ^ {\prime}$ be the graph we get after removing from $G$ all edges on some separator path in $X$. Still for $r \in X$ vertices $r_k$ and $r_l$ share a biconnected component if and only if $|k - l| = 2$. Our problem is that if $u$ is $r$ strongly-connected to $r_{i - 2}$, it is no longer obvious whether $u$ share a biconnected component with both $r_i$ and $r_{i - 2}$ in $G ^ {\prime}$, since some edges not on $r$ are removed. Fortunately, in this case, we have the following Theorem:

\begin{theorem} \label{lrtheo}
    If $u$ is strongly $r$-connected to $r_{i - 2}$ and there exists an edge between $u$ and $r_i$ in $G ^ {\prime}$, then $u$, $r_i$ and $r_{i - 2}$ share a biconnected component in $G ^ {\prime}$.
\end{theorem}

To prove Theorem \ref{lrtheo}, we first introduce the following lemma:

\begin{lemma} \label{lemmanoshare}
    If $r$ and $q$ are useful separator paths, then for $0 \le i < |r| - 2$, the edge between $r_i$ and $r_{i + 2}$ (which exists due to Theorem \ref{theoadjacent}) is not on $q$.
\end{lemma}

\begin{proof} [Proof of Lemma \ref{lemmanoshare}]
    If the edge between $r_i$ and $r_{i + 2}$ belongs to another useful separator path $q$. Then $q$ must contain either the edge between $r_i$ and $r_{i + 1}$ or the edge between $r_{i + 1}$ and $r_{i + 2}$. Otherwise two adjacent vertices on $q$ are strongly $q$-connected, violating Theorem \ref{pathseparatorprop1}. Without loss of generality, suppose $q$ contains the edge between $r_i$ and $r_{i + 1}$ and $r_{i + 1} = q_j$, $r_i = q_{j + 1}$ and $r_{i + 2} = q{j + 2}$, then since $r$ is useful:    
    
    \begin{align*}
        L(q_j, q_{j + 2})   &= L(r_{i + 1}, r_{i + 2}) \\
                            &< L(r_{i + 1}, r_{i + 2}) + L(r_{i + 1}, r_{i}) \\
                            &< L(r_{i}, r_{i + 2}) \\
                            &< L(q_{j + 1}, q_{j + 2}) \\
                            &< L(q_{j + 1}, q_{j + 2}) + L(q_{j + 1}, q_{j})
    \end{align*}
    
    which contradicts the fact that $q$ is useful.
\end{proof}

\begin{proof} [Proof of Theorem \ref{lrtheo}]
    Since $r_i$ and $r_{i - 2}$ are adjacent due to Theorem \ref{theoadjacent}, it suffices if we can show there is a simple path between $u$ and $r_{i - 2}$ that does not visit another vertex in $r$. Consider a simple path $p$ in $G$ (not $G ^ {\prime}$) between $u$ and $r_{i - 2}$ that does not visit another vertex in $r$. We show that we can adjust $p$ so that $p$ does not visit any edge on a separator path in $X$, and then $p$ will exist in $G ^ {\prime}$. For $r ^ {\prime} \in X$, if $p$ does not visit any vertex on $r ^ {\prime}$, $p$ obviously does not visit an edge on $r ^ {\prime}$. Otherwise, let the first vertex on $r ^ {\prime}$ that $p$ visits be $(r ^ {\prime})_u$ and the last vertex be $(r ^ {\prime})_v$. Since the edge between $u$ and $r_i$ is not on $r ^ {\prime}$ and from Theorem \ref{nosharee}, no edges on $r_{i - 2, i}$ are on $r ^ {\prime}$, $(r ^ {\prime})_u$ and $(r ^ {\prime})_v$ are weakly $(r ^ {\prime})$-connected. It follows from Lemma \ref{lemmapathseparator1} that $|u - v|$ is even. From Theorem \ref{theoadjacent} and Lemma \ref{lemmanoshare} we can get a path between $(r ^ {\prime})_u$ and $(r ^ {\prime})_v$ that is not on any separator path in $X$ using the edges between $(r ^ {\prime})_k$ and $(r ^ {\prime})_{k + 2}$ ($\min\{u, v\} \le k < \max\{u, v\}$ and $k - u$ is even). Therefore we can replace the part of $p$ from $(r ^ {\prime})_i$ to $(r ^ {\prime})_j$ with this path. It suffices to do the adjustment for all $r ^ {\prime} \in X$.
\end{proof}

Now, if the edge between $u$ and $r_i$ exists in $G ^ {\prime}$, we can readily check if $u$ and $r_{i - 2}$ are strongly $r$-connected. Otherwise, there must be some $r ^ {\prime} \in X$ where $r_i = (r ^ {\prime})_j$ and $u = (r ^ {\prime})_k$ where $|k - j| = 1$. Due to theorem \ref{nosharev}, $j = 0$ or $j = |r ^ {\prime}| - 1$. Without loss of generality let $j = 0$, and then $k = 1$. Consider $(r ^ {\prime})_2$. If $(r ^ {\prime})_2 \notin r$, then $u$ and $r_{i - 2}$ are strongly $r$-connected if and only if $r ^ {\prime})_2$ and $r_{i - 2}$ are strongly $r$-connected. Since $r_i = (r ^ {\prime})_0$ and $(r ^ {\prime})_2$ are adjacent and by Lemma \ref{lemmanoshare} the edge between them is in $G ^ {\prime}$, we can apply \ref{lrtheo}. If $(r ^ {\prime})_2 \in r$, then $u$ and $r_{i - 2}$ are strongly $r$-connected if and only if $r ^ {\prime})_2$ and both $r_i$ and $r_{i - 2}$ are strongly connected, which is only possible if $(r ^ {\prime})_2 = r_{i - 2}$.

\subsection{Computation of Bad Vertices, $S$-Separators Paths and $T$-Separators Paths} \label{badlr}

In this sub-section, we will introduce procedures that calculate all the bad vertices, and the $S$-separator paths and $T$-separator paths given $P$, the shortest path between $S$ and $T$ in $G ^ {\prime}$. Both procedures will be based on the calculation of the \textit{block-cut tree}, a data structure based on biconnected components of undirected graphs, which can be done in $O(n + m)$ time for a graph with $n$ vertices and $m$ edges due to the famous algorithm by John Hopcroft and Robert Tarjan \cite{tarjan}.

\subsubsection{Computation of Bad Vertices}

Per definition, a bad vertex is the middle point of a traversable separator path of length 2. If $r$ is a separator path of length 2. Then from Theorem \ref{pathseparatorprop1} one can see that $r_1$ is an articulation vertex of $G$, and $r_0$, $r_1$ and $r_2$ share the same biconnected component. Moreover, if $x$ is a vertex in a biconnected component $C$ such that the degree of $x$ within $C$ is two. Then the two edges associated with $x$ in $C$ constitute a separator path of length 2. $x$ will be bad vertex as long as $r$ is traversable.

Consider the block-cut tree $\tau$ of $G$. Let $S_{\tau}$ be the vertex corresponding $S$ itself if $S$ is an articulation point and to the biconnected component $S$ is in otherwise, and let $T_{\tau}$ be that for $T$. Then one can easily verify that $r$ is traversable if and only if $C$ is on the path from $S_{\tau}$ to $T_{\tau}$ on $\tau$. Therefore, $\textrm{BAD}_G$ can be computed using the following sub-procedure:

\begin{algorithm}
    \caption{Computation of $\textrm{BAD}_G$}
    \begin{algorithmic}[1]
        \Require Chordal Graph $G = \langle V, E\rangle$, $S, T \in V (S \ne T)$
        \Ensure $\textrm{BAD}_G$, the set of bad vertices
        \State Build the block-cut tree $\tau$ of $G$.
        \If {$S$ is an articulation vertex}
            \State $S_{\tau} \gets $the vertex corresponding to $S$ on $\tau$
        \Else
            \State $S_{\tau} \gets $the vertex corresponding to the biconnected component $S$ is in on $\tau$
        \EndIf
        \If {$T$ is an articulation vertex}
            \State $T_{\tau} \gets $the vertex corresponding to $T$ on $\tau$
        \Else
            \State $T_{\tau} \gets $the vertex corresponding to the biconnected component $T$ is in on $\tau$
        \EndIf
        \State $\textrm{BAD}_G \gets \emptyset$
        \For {Vertex $v$ between $S_{\tau}$ and $T_{\tau}$ on $\tau$ that corresponds to a biconnected component}
            \State $C \gets $the biconnected component $v$ corresponds to
            \For {$u \in C$}
                \If {$u$ has a degree of exactly 2 in $C$}
                    \State $\textrm{BAD}_G \gets \textrm{BAD}_G \cap \{u\}$
                \EndIf
            \EndFor
        \EndFor \\
        \Return $\textrm{BAD}_G$
    \end{algorithmic}
\end{algorithm}

\newpage

\subsubsection{Computation of $S$-Separator Paths and $T$-Separator Paths on $P$}

Since $S$-separator paths and $T$-separator paths are symmetric, it suffices to design an algorithm that calculates $S$-separator paths. We can split $S$-separator paths into two categories by the parity of the index of $r_2$ on $P$. An $S$-separator path is called \textit{even} if $\textrm{INDEX}(P, r_2)$ is even, and \textit{odd} if $\textrm{INDEX}(P, r_2)$ is odd. We will calculate all the odd $S$-separator paths, and the even ones can be calculated similarly.

In the previous section, we were able to pick out a length-2 separator path since all edges on it are associated with its midpoint. An odd $S$-separator path can contain more than 3 vertices and no single vertex is associated with all the edges on the path. Fortunately, we can design a way to merge $r_1, r_3, \cdots$ into a fat vertex without including any of $r_0, r_2, \cdots$, so that all edges on the path will be associated with the fat vertex. Here by merging $u$, $v$ into a fat vertex we mean replacing $u$ and $v$ with a new vertex $w$ such that for every edge between $u$ and $x$ or $v$ and $x$ $(x \notin \{u, v\})$ before merging, there is an edge between $w$ and $x$ with the same length as that edge after merging. We also use a shorthand: if we merge $a$ and $b$ into a fat vertex $x$ first and merge $b$ and $c$ thereafter, we actually merge $x$ and $c$.

To merge $r_1, r_3, \cdots$ into a fat vertex, we first go through $0 \le i < |P| - 2$ where $i$ is even, and merge $P_i$ with $P_{i + 2}$ if they are adjacent in $G$. This will make sure that $r_3, r_5, \cdots$ are merged into a fat vertex, since all these vertices are on $P$ with even indices and there are edges inter-connecting them due to Theorem \ref{theoadjacent}, and none of $r_0, r_2, \cdots$ will be merged into this fat vertex due to Theorem \ref{pathseparatorprop1}.

It remains to merge $r_1$ and $r_3$. If $r_1 \in P$, then since $\textrm{INDEX}(P, r_1) < \textrm{INDEX}(P, r_2)$ per definition, the only way to make sure that $r_1$ and $r_2$ are not strongly $r$ connected is for $r_{1, 2}$ to be on $P$ and therefore $\textrm{INDEX}(P, r_1) = \textrm{INDEX}(P, r_2) - 1$ which is even. This would mean that $r_1$ is already merged into the fat vertex. If $r_1 \notin P$, since $r_1$ and $r_3$ are adjacent and $\textrm{INDEX}(r_3)$ is odd, a way we can do this is to merge all $u \notin P$ with $v \in P$ where $u$ and $v$ are adjacent and $\textrm{INDEX}(v)$ is odd. This will not merge $r_1$ with $r_2$ since $\textrm{INDEX}(r_2)$ is even. However, this method can merge $r_1$ with $r_0$ when $r_0 \in P$ and $\textrm{INDEX}(P, r_0)$ is odd. To prevent this, note that if $r_0 \in P$ and we let $j = \textrm{INDEX}(r_0)$, then for no $k < j$ can $P_k$ and $r_1$ be adjacent --- otherwise $r_1$ and $r_0$ will be strongly $r$-connected. Therefore, we do not merge $u$ with $v$ if $v$ has the smallest index among all $v \in P$ adjacent to $u$. This will not prevent $r_1$ from merging with $r_3$ since $\textrm{INDEX}(P, r_2) < \textrm{INDEX}(P, r_3)$ and $r_2$ is adjacent to $r_1$.

Now that the edges on a the odd $S$-separator path are associated with a single fat vertex, we can pick out these separator paths by checking for each biconnected component, if the set of edges associated with each vertex inside that biconnected component corresponds to a normal separator path. Our sub-procedure will be as follows:

\begin{algorithm}
    \caption{Computation of Odd $S$-Separators Paths}
    \begin{algorithmic}[1]
        \Require Chordal Graph $G = \langle V, E\rangle$, $S, T \in V (S \ne T)$, path $P$ from $S$ to $T$
        \Ensure $X_{\textrm{SODD}}$, the set of odd $S$-separator paths
        \State $\textrm{TO\_MERGE} \gets \emptyset$
        \For {$0 \le i < |P| - 2$}
            \If {$i$ is odd}
                \If {$P_i$ and $P_{i + 2}$ are adjacent}
                    \State Add $(P_i, P_{i + 2})$ to $\textrm{TO\_MERGE}$.
                \EndIf
            \EndIf
        \EndFor
        \For {Vertex $u \notin P$}
            \State $\textrm{MIN\_INDEX} \gets \infty$
            \For {Vertex $v$ such that $u$ and $v$ are adjacent}
                \If {$v \in P$}
                    \State $\textrm{MIN\_INDEX} \gets \min\{\textrm{MIN\_INDEX}, \textrm{INDEX}(v, P)\}$
                \EndIf
            \EndFor
            \For {Vertex $v$ such that $u$ and $v$ are adjacent}
                \If {$v \in P$ \AND $\textrm{INDEX}(P, v) \ne \textrm{MIN\_INDEX}$}
                    \State Add $(u, v)$ to $\textrm{TO\_MERGE}$.
                \EndIf
            \EndFor
        \EndFor
        \State $G_{\textrm{MERGE}} \gets G$
        \For {$(u, v) \in \textrm{TO\_MERGE}$}
            \State Merge $u$ and $v$ in $G_{\textrm{MERGE}}$.
        \EndFor
        \State Build the block-cut tree $\tau$ of $G_{\textrm{MERGE}}$.
        \State $X_{\textrm{SODD}} \gets \emptyset$
        \For {Vertex $v$ on $\tau$ that corresponds to a biconnected component}
            \State $C \gets $the biconnected component $v$ corresponds to \Comment{Including cut vertices}
            \For {$u \in C$}
                \State $R \gets $edges associated with $u$ inside $C$
                \If {$R$ corresponds to an odd $S$-separator path $r$ in $G$}
                    \State $X_{\textrm{SODD}} \gets X_{\textrm{SODD}} \cap \{r\}$ 
                \EndIf
            \EndFor
        \EndFor \\
        \Return $X_{\textrm{SODD}}$
    \end{algorithmic}
\end{algorithm}

\newpage

We still need to check if a set of edges $R$ in $G$ corresponds to an odd $S$-separator path in $O(|R|)$ time. Firstly, we check if $R$ corresponds to a path. It is not hard to see that $R$ is a minimal separating set of edges, so if $R$ corresponds to a path, the path is to a separator path $r$. Secondly, one can determine the direction of the path based on the indices of the vertices $r$ shared with $p$ --- recall that for any two vertices $u$ and $v$ shared by $r$ and $P$, $\textrm{INDEX}(r, u) < \textrm{INDEX}(r, v)$ if and only if $\textrm{INDEX}(P, u) < \textrm{INDEX}(P, v)$. If the constraint on indices can't be satisfied in either direction, then by definition $r$ is not an $S$-separator path. After $r$ has been fully determined, we can then check whether $r_{2, |r| - 1}$ is on $P$. Thirdly, we check for the usefulness constraint per definition, and finally we need to check for the traversability constraint, we have the following theorem.

\begin{theorem} \label{adjacent}
    If $r$ is a useful separator path such that $r_{2, |r| - 1}$ is on $P$ and for any two vertices $u$ and $v$ shared by $r$ and $P$, $\textrm{INDEX}(r, u) < \textrm{INDEX}(r, v)$ if and only if $\textrm{INDEX}(P, u) < \textrm{INDEX}(P, v)$. Then let $j = \textrm{INDEX}(P, r_2)$ and $r$ is traversable if and only if $r_0$ is equal to or adjacent to $P_k$ for some $k < j$.
\end{theorem}

\begin{proof} [Proof of Theorem \ref{adjacent}]
    If $r_0$ is equal to or adjacent to $P_k$ for some $k < j$, we can obviously find a path from $S$ to $T$ that contains $r$, and therefore $r$ is traversable.
    
    If $r$ is traversable and therefore normal, then if $r_0 \in P$ the theorem is obviously true. If $r_0 \notin P$, then from Corollary \ref{corollary:firstvisit} $r_2$ is the first vertex on $r$ on $P$ and $r_1 \notin P$ from the constraint on indices, and the indices of the vertices of $r$ on $P$ form an interval. Since $r$ is traversable, there exists a path from $S$ to $r_0$ that does not visit a vertex on $r$ before $r_0$. If the path last leaves $P$ at $P_l$, then $l < j (= \textrm{INDEX}(P, r_2))$ since otherwise $l > \textrm{INDEX}(P, r_{|r| - 1})$ and $r_0$ and $r_{|r| - 1}$ will be strongly $r$-connected, contradicting $|r| \ge 3$. Now, let $k$ be the largest index such that $k < j$ and there exists a path between $P_k$ and $r_0$ that does not visit a vertex on $r$. Find the shortest of such path $p$. If $p$ contains at least two edges, then consider the simple cycle consisting of $p$, $P_{k, j}$ and the edge between $r_2$ and $r_0$. This cycle has a length of at least four edges and one can verify that it is impossible for any chord to exist, which contradicts the fact that $G$ is chordal. If $p$ contains only one edge, then $P_k$ is adjacent to $r_0$.
\end{proof}

With Theorem \ref{adjacent}, we can pre-compute for each $u \notin P$, the minimum $k$ such that $P_k$ is adjacent to $u$. With this we can check in $O(|R|)$ time whether $R$ corresponds to an odd $S$-separator path. Since the total number of edges ever involved in $R$ in the sub-procedure is $O(|E|)$. The entire sub-procedure runs in $O(|V| + |E|)$ time. 

However, there are still two issues left to address. Firstly, if we merge vertices using typical disjoint-set data structures, it would take $O(\alpha(|V|))$ time asymptotically per action, yielding a total time complexity of $O(|V| + |E|\alpha(|V|))$. We can improve this to $O(|V| + |E|)$: we can move all the actions offline by creating a graph representing all the merges, and a fat vertex will correspond to a connected component of the graph.

Secondly, checking whether a separator path is useful involves looking up the length of edges between some given pairs of vertices. Although these look-ups can easily be done in $O(1)$ per query if the graph is stored with an adjacency matrix, for many other ways one stores a graph (e.g. linked lists), the most obvious way of doing these operations in $O(1)$ time would be to use a hash table \footnote{Assuming the RAM model.}, which is somewhat awkward since hash tables inherently introduce randomness. An alternative way to hashing is by doing the look-ups offline. We can use a bucket for each vertex and store the queries into the bucket for either of the associated vertex. After that we can deal with each bucket alone. The look-ups for edges associated with a given vertex can be done in $O(1)$ time per query by using a 1-D array of size $|V|$All the look-ups can therefore be done in $O(|V| + Q)$ time offline where $Q$ is the number of queries, without the need for a hash table. 

\subsection{Computation of Normal Separator Paths Contained in a Simple Path} \label{sepcontained}

During the main procedure we find a path $P_0$ from $S$ to $T$ in $G_0$ and we want to find all separator paths contained in $P_0$. 

The sub-procedure will work as follows: firstly, let $G_3$ be the graph with all the edges on $P_0$ removed. We then find all the connected components in $G_3$. Then we have the following theorem:

\begin{theorem} \label{theosepcontained}
    If $(P_0)_{i, j} (0 \le i < j < |P_0|)$ is a separator path, then in $G_3$ all the vertices with odd indices on $(P_0)_{i, j}$ belong to a connected component $C_o$, and all the vertices with even indices on $(P_0)_{i, j}$ belong to a separate connected component $C_e$. Reversely, for $0 \le i < j < |P_0|$, if in $G_3$ all the vertices with odd indices on $(P_0)_{i, j}$ belong to a connected component $C_o$, and all the vertices with even indices on $(P_0)_{i, j}$ belong to a separate connected component $C_e$, and to no $(P_0)_{x, y} (0 \le x < y < |P_0|)$ where $(P_0)_{i, j} \subset (P_0)_{x, y}$ such condition applies (i.e. $(P_0)_{i, j}$ is maximal), then $(P_0)_{i, j}$ is a separator path.
\end{theorem}

\begin{proof} [Proof of Theorem \ref{theosepcontained}]
    Firstly, if $(P_0)_{i, j}$ is a separator path, then straightforwardly from Theorem \ref{pathseparatorprop1} we know all the vertices with odd indices belong to the same connected component in $G_3$ and so do the ones with even indices. The two connected components are different since otherwise $(P_0)_{i, j}$ is not separating.
    
    Secondly, if in $G_3$ all the vertices with odd indices on $(P_0)_{i, j}$ belong to a connected component $C_o$, and all the vertices with even indices on $(P_0)_{i, j}$ belong to a separate connected component $C_e$, and $(P_0)_{i, j}$ is not a separator path. Then consider a minimal subset $D$ of edges on $P_0$ such that after removing $D$ from $G$, vertices in $C_o$ and vertices in $C_e$ are disconnected. From Lemma \ref{minsepconnected} we know there must be $0 \le k < l < |P_0|$ such that $D$ is the set of edges on $(P_0)_{k, l}$. If $(P_0)_{i, j} \not\subset (P_0)_{k, l}$, then a edge not in $D$ will be between $C_o$ and $C_e$ and removing $D$ from $G$ will not disconnect $C_o$ and $C_e$, which is a contradiction. Therefore $(P_0)_{i, j} \subset (P_0)_{k, l}$.
\end{proof}

With Theorem \ref{theosepcontained} we can design an algorithm based on the \textit{two-pointer} technique. We iterate through the index $j$ in the increasing order and maintain the minimum $i$ where $(P_0)_{i, j}$ satisfies the condition in the theorem. Every time we need to change $i$, $(P_0)_{i, j - 1}$ is maximal and is a separator path. If the minimum $i$ for $j = |P_0| - 1$ is less than $j$, $(P_0)_{i, j}$ is maximal and is a separator path. To make sure the separator path is normal, we check if it is useful. The pseudocode is as follows:

\begin{algorithm}
    \caption{Computation Normal Separator Paths on a Simple Path}
    \begin{algorithmic}[1]
        \Require Chordal Graph $G = \langle V, E\rangle$, a simple path $P_0$
        \Ensure $X_{\textrm{EXTRA}}$, the set of separator paths contained in $P_0$.
        \State $i \gets 0$
        \State $G_3 \gets $ $G$ with edges on $P_0$ removed.
        \For {$0 \le i < |P_0|$}
            \State $\textrm{BEL}_i \gets $the connected component $(P_0)_i$ belongs to in $G_3$
        \EndFor
        \State $X_{\textrm{EXTRA}} \gets \emptyset$
        \For {$j$ from $1$ to $|P_0| - 1$}
            \If {$j > 1$ \AND $\textrm{BEL}_j \ne \textrm{BEL}_{j - 2}$}
                \If {($i < j - 1$) \AND ($(P_0)_{i, j - 1}$ is useful)}
                    \State $X_{\textrm{EXTRA}} \gets X_{\textrm{EXTRA}} \cup \{(P_0)_{i, j - 1}\}$
                \EndIf
                \State $i \gets j - 1$
            \EndIf
            \If {$\textrm{BEL}_j = \textrm{BEL}_{j - 1}$}
                \State $i \gets j$ \Comment{It is not hard to see that $i$ always equals $j - 1$ before this line, and no new separator path is found}
            \EndIf
        \EndFor
        \If {$i < |P_0| - 1$}
            \State $X_{\textrm{EXTRA}} \gets X_{\textrm{EXTRA}} \cup \{(P_0)_{i, |P_0| - 1}\}$
        \EndIf \\
        \Return $X_{\textrm{EXTRA}}$
    \end{algorithmic}
\end{algorithm}

\newpage

\section{Correctness}

\subsection{Correctness of Normal Separator Paths} \label{correctnormal}

We will show that any separator paths contained in a path $P$ that can be produced by $\textrm{AVOID}(X)$ for some $X$ is indeed useful, and therefore normal.

\begin{lemma} \label{replacewithedge}
   Let $p$ be a simple path and let $0 \le i < j < |p|$ be such that $j - i > 1$ and $p_i$ and $p_j$ are adjacent, let $p ^ {\prime}$ be the path we get from $p$ after replacing $p_{i, j}$ by the edge between $p_i$ and $p_j$. For all $u, v \in V$, if $u$ and $v$ are weakly $p$-connected, then $u$ and $v$ are weakly $p ^ {\prime}$-connected.
\end{lemma}

\begin{proof}
    For $u, v \in V$ such that $u$ and $v$ are weakly $p$-connected., for any path between $u$ and $v$ that do not share an edge with $p$, if the path goes through the edge between $p_i$ and $p_j$, replace that edge with $p_{i, j}$ and we find a path between $u$ and $v$ that do not share an edge with $p ^ {\prime}$. Therefore $u$ and $v$ are weakly $p ^ {\prime}$-connected.
\end{proof}

\begin{theorem} \label{correctnormaltheo}
    If $r$ is a separator path that is not useful, then for all simple path $p$ such that $r \subset p$, there exists another simple path $p ^ {\prime}$ of shorter length such that for all separator path $r ^ {\prime} \subset p ^ {\prime}$, $r ^ {\prime} \subset p$.
\end{theorem}

\begin{proof}
    If $\|r_{i, i + 2}\| > L(r_i, r_{i + 2})$ for some $0 \le i < |r| - 2$, let $p ^ {\prime}$ be $p$ after replacing the $r_{i, i + 2}$ with the edge between $r_i$ and $r_{i + 2}$. $p ^ {\prime}$ will have a shorter length. For separator path $r ^ {\prime} \subset p ^ {\prime}$, if $r ^ {\prime} \notin p$, suppose the edge between $(r ^ {\prime})_i$ and $(r ^ {\prime})_{i + 1}$ is not on $p$. $(r ^ {\prime})_i$ and $(r ^ {\prime})_{i + 1}$ are not weakly $r ^ {\prime}$-connected according to Lemma \ref{lemmapathseparator1}, and therefore not weakly $p ^ {\prime}$ connected, but are weakly $p$-connected. This is a contradiction of Lemma \ref{replacewithedge}.
\end{proof}

Therefore, if $P$ contains a separator path that is not useful, from Theorem \ref{correctnormaltheo} we know we can make $P$ shorter while still avoiding $X$, which contradicts the fact that $P$ is the shortest.

\subsection{Correctness of the Set $X$} \label{correctx}

In this sub-section, we show that the set $X$ computed in the main algorithm is indeed such that the shortest path that avoids $X$ does not contain a new normal separator path not in $X$. 

Let $P_{\textrm{AVOID}}$ be the path produced by $\textrm{AVOID}(X)$. Let $X_n$ be the set of normal separator paths that $P_{\textrm{AVOID}}$ contains. Obviously, $X \cap X_n = \emptyset$. Our goal is to show that $X_n = \emptyset$. Now, suppose we have a hypothetical $n \in X_n$, and we want to prove such $n$ actually does not exist and therefore conclude $X_n = \emptyset$. If $P_i$ is an inner vertex of $n$, we show that we have two vertices $a = P_j$ and $b = P_k$, where $j < i < k$, such that we have the following four conditions:

\begin{enumerate}
    \item For any $i < x < j$, if $P_x$ is the inner vertex of an $S$-separator path $s ^ {\prime}$, then $(s ^ {\prime})_{|s ^ {\prime}| - 2} \notin P_{\textrm{AVOID}}$; if $P_x$ is the inner vertex of an $T$-separator path $t ^ {\prime}$, then $t ^ {\prime})_1 \notin P_{\textrm{AVOID}}$
    \item No inner vertex of a separator path $r \in X_{\textrm{EXTRA}}$ is contained in $P_{i, j}$
    \item $a = s_{|s| - 3}$ for an $S$-separator path $s$, and $b = t_{2}$ for a $T$-separator path $t$. 
    \item $a, b \in P_{\textrm{AVOID}}$, $a$ is the last inner vertex of $s$ that $P_{\textrm{AVOID}}$ visits, and $b$ is the first inner vertex of $t$ that $P_{\textrm{AVOID}}$ visits.
\end{enumerate}

We show that the indices of the inner vertices of a $S$-separator path or a $T$-separator path on $P$ must be an interval. Note that if $r$ is an $S$-separator path and $r_1 \in P$, $r_1$ must be adjacent to $r_2$ on $P$ since otherwise $r_1$ and $r_2$ are strongly $r$-connected. Similar reasoning applies for $T$-separator paths. We can define an order on the set $X_{ST}$ according to the order their corresponding intervals appear on $P$. For distinct $x, y \in X_{ST}$, $x < y$ if the interval for $x$ appears closer to the $S$-side than the interval for $y$, that is, $\forall P_u \in \textrm{INNER(x)}$ and $\forall P_v \in \textrm{INNER(y)}$, $u < v$. Figure \ref{fig:main5} shows the intervals and the ordering in the example graph we have shown before.

\begin{figure}
    \includegraphics[width=1\textwidth]{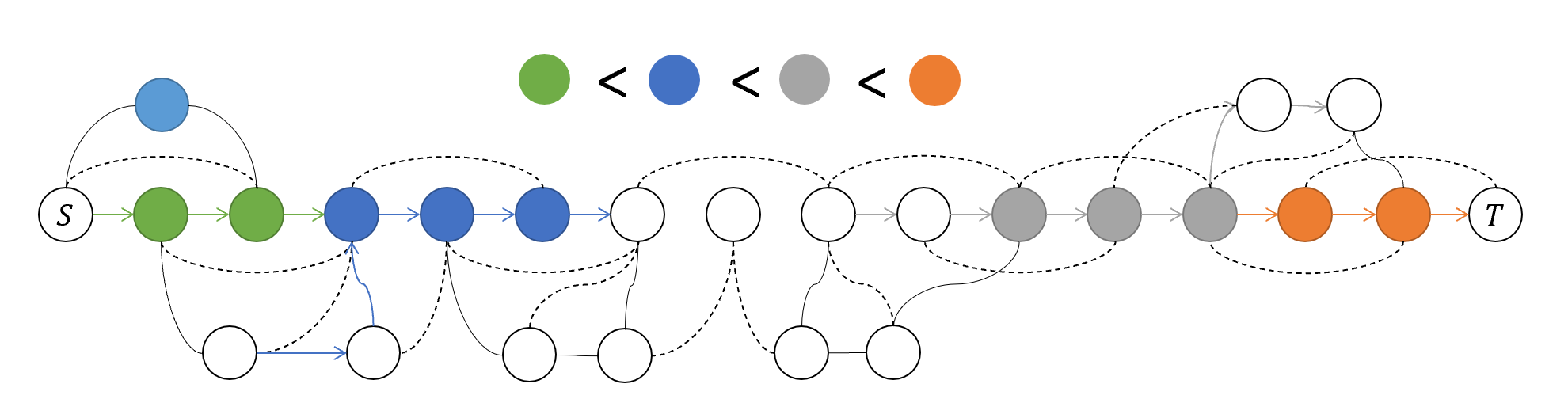}
    \caption{Intervals for $X_{ST}$}
    \label{fig:main5}
\end{figure}

We have the following lemma:

\begin{lemma} \label{lemma:storder}
    Let $p$ be a simple path from $S$ to $T$.
    
    If $r$ is a normal separator path such that the indices of the set of inner vertices of $r$ on $P$ is an interval, and $P_u$ is an inner vertex of $r$ on $P$, then:
    
    \begin{itemize}
        \item After $p$ visits the an inner vertex of $r$ for the last time, $p$ does not visit an inner vertex of any normal separator path $x$ such that there exists $P_k \in \textrm{INNER}(x)$ (i.e. an inner vertex of $x$ on $P$) such that $k < u$.
        \item Before $p$ visits an inner vertex of $r$ for the first time, $p$ does not visit an inner vertex of any normal separator path $x$ such that there exists $P_k \in \textrm{INNER}(x)$ such that $k > u$.
    \end{itemize}
    
    Consequently, if $r \in X_{ST}$:
    
    \begin{itemize}
        \item After $p$ visits an inner vertex of $r$ for the last time, $p$ does not visit an inner vertex of any separator path $x \in X_{ST}$ where $x < r$.
        \item Before $p$ visits an inner vertex of $r$ for the first time, $p$ does not visit an inner vertex of any separator path $x \in X_{ST}$ where $r < x$.
    \end{itemize}
\end{lemma}

\begin{proof}  
    Suppose after $p$ visits the an inner vertex of $r$ for the last time, $p$ visits an inner vertex $v$ of any normal separator path $x$ such that there exists $P_k \in \textrm{INNER}(x)$ such that $k < u$. Then $v$ are $S$ are strongly $r$-connected. Since $S$ and $r_0$ are strongly $r$-connected, $v$ and $r_0$ are strongly $r$-connected. Since $p$ does not visit any inner vertex of $r$ after $v$, $v$ and $r_{|r| - 1}$ are strongly $r$-connected and therefore $r_0$ and $r_{|r| - 1}$ are strongly $r$-connected, which is impossible since $|r| > 3$.
    
    The other direction is similar.
\end{proof}

We argue that both $P_{\textrm{AVOID}}$ and $P_0$ must visit $a$ and $b$, and $a$ before $b$. Moreover, let $p_{\textrm{AVOID}}$ be the segment of $P_{\textrm{AVOID}}$ from $a$ to $b$, and $p_0$ be the segment of $P_0$ from $a$ to $b$, and we argue that $p_{\textrm{AVOID}} = p_0$. We will show it by contradiction both when $\|p_{\textrm{AVOID}}\| < \|p_0\|$ and $\|p_{\textrm{AVOID}}\| > \|p_0\|$. To do this, we will argue that $p_{\textrm{AVOID}}$ is also a path in $G_0$. Therefore if $\|p_{\textrm{AVOID}}\| < \|p_0\|$ one can replace, in $P_0$, the segment $p_0$ by $p_{\textrm{AVOID}}$ and find a shorter path than $P_0$ between $S$ and $T$ in $G_0$, which is a contradiction. We also argue that in $P_{\textrm{AVOID}}$, if we replace $p_{\textrm{AVOID}}$ by $p_0$ and get $P_{\textrm{NEW}}$, $P_{\textrm{NEW}}$ still avoids $X$, and therefore if $\|p_{\textrm{AVOID}}\| < \|p_0\|$, $P_{\textrm{NEW}}$ becomes a shorter $X$-avoiding path in $G$ between $S$ and $T$ in $G$ than $P_{\textrm{AVOID}}$, which is a contradiction.

We show that $P_0$, the shortest path between $S$ and $T$ in $G_0$, must visit both $a$ and $b$, and $a$ before $b$. Recall that $a = s_{|s| - 3}$ and $b = t_2$. Note that in $G_0$, $t_{2}$ is an articulation vertex. To show this, it suffices to show that any simple path between $S$ and $T$ visits $t_2$. Consider Corollary \ref{corollary:firstvisit}. If the first vertex of $t$ that $P_0$ visits is not $t_2$, that vertex is $t_0$. It follows from Corollary \ref{noreverse} that the path visits some other vertex of $t$ after visiting $t_0$. However, since the edge between $t_0$ and $t_1$ does not exist in $G_0$ and the only other vertex of $t$ that is strongly $t$-connected to $t_0$ is $t_2$, the path must also visit $t_2$. Note that this reasoning also shows that $b (t_2)$ is the first inner vertex of $t$ that $P_0$ visits. Similarly, $a (s_{|s| - 3})$ is an articulation vertex and is the last inner vertex that $P_0$ visits. Therefore $P_0$ must visit both $a$ and $b$ and since $s < t$, if $b$ is visited before $a$, we immediately violate Lemma \ref{lemma:storder}. Therefore, $a$ is visited before $b$.

By definition $P_{\textrm{AVOID}}$ visits both $a$ and $b$. It remains to show that $a$ is visited before $b$. Suppose $b$ is visited before $a$. Let $p_{ba}$ be the segment on $P_{\textrm{AVOID}}$ from $b$ to $a$, and $p_{ab}$ be the reverse of $p_{ba}$. Since $p_{ba}$ contains $n$, $p_{ab}$ contains the reverse of $n$. Consider the path concatenated by $P_{0, i}$ (the segment from $S$ to $a$), $p_{ab}$ and $P_{j, |P| - 1}$ (the segment from $b$ to $T$). This is a path that contains $p_{ab}$ and therefore the reverse of $n$. The reverse of $n$ is also traversable and both $n$ and its reverse are normal separator paths. This violates Corollary \ref{noreverse}.

Now we show that $p_{\textrm{AVOID}}$, the segment of $P_{\textrm{AVOID}}$ from $a$ to $b$, is a path in $G_0$. It suffices to show that it does not contain the tail of an $S$-separator path and the head of a $T$-separator path. Let $x$ be an $S$-separator path and we want to show that its tail is not on $p_{\textrm{AVOID}}$. The case where $x \in \{s, t\}$ is trivial. The case where $s < x < t$ is straightforward from the first condition. The case where $x < s$ or $x > t$ is straightforward from Lemma \ref{lemma:storder}. The case for $T$-separator paths is similar.

Now we show that $P_{\textrm{NEW}}$, the path we get after replacing $p_{\textrm{AVOID}}$ by $p_0$ in $P_{\textrm{AVOID}}$, still avoids $X$. Consider $x \in X$. The first case is if $x \in X_{ST}$. In this case if $x < s$ or $x > t$, then from Lemma \ref{lemma:storder} no inner vertices of $x$ and therefore no edges on $x$ are on $p_{\textrm{AVOID}}$. Since $P_{\textrm{AVOID}}$ does not contain $x$, neither does $P_{\textrm{NEW}}$. If $x \in \{s, t\}$, $p_{\textrm{AVOID}}$ also does not contain an edge on $x$ due to the fourth condition. If $s < x < t$, without loss of generality suppose $x$ is an $S$-separator path. From the first condition, $x_{|x| - 2} \notin P_{\textrm{AVOID}}$ and therefore $x_{|x| - 2, |x| - 1} \notin P_{\textrm{AVOID}}$. Since $p_0$ is a path in $G_0$, $x_{|x| - 2, |x| - 1} \notin p_0$. Therefore $x_{|x| - 2, |x| - 1} \notin P_{\textrm{NEW}}$ and $P_{\textrm{NEW}}$ does not contain $x$. If $x \in X_{\textrm{EXTRA}}$, we argue that no inner vertex of $x$ and therefore no edge on $x$ is on $p_0$, which will finish the proof. Suppose an inner vertex $u$ of $x$ is on $p_0$. From Corollary \ref{corollary:innervisit}, we know an inner vertex $v$ of $x$ is on $P$, and from the second condition either $v < s$ or $v > t$, both would violate Lemma \ref{lemma:storder}.
  
It remains to show that such $a$ and $b$ exist. For a separator path $r$, a path $p$ is called $r$-free if $p$ does not visit a bad vertex and, except at its endpoints, $p$ does not visit a vertex of $r$. A normal separator $r$ is said to be \textit{truly useful} if, in addition to being useful, for all $1 \le i < |r| - 3$, the shortest $r$-free path between $r_i$ and $r_{i + 2}$ is has a length greater than $\|r_{i, i + 2}\|$. Let $X ^ {+} = X \cup X_n$. We have the following:

\begin{lemma} \label{trulyuseful}
    All separator paths in $X ^ {+} \backslash X_{ST}$ are truly useful.
\end{lemma}

To prove Lemma \ref{trulyuseful}, we introduce the following lemma, to be proven at the end of this sub-section.

\begin{lemma} \label{lemma:noinner}
     Let $r$ be a normal separator. For all $1 \le i < |r| - 3$, any $r$-free path between $r_i$ and $r_{i + 2}$ does not visit an inner vertex of another normal separator path.
\end{lemma}

\begin{proof} [Proof of Lemma \ref{trulyuseful}]
    A separator path $r \in X ^ {+} \backslash X_{ST}$ is on a path $p$ which is either the shortest $X_{ST}$-avoiding path or the shortest $X$-avoiding path. If $r$ is not truly useful, then there exists $i \le |r| - 3$ such that in $p$ we can replace $r_{i, i + 2}$ with the shortest $r$-free path between $r_i$ and $r_{i + 2}$ and make $p$ shorter. Due to Lemma \ref{lemma:noinner}, no $r$-free path between $r_i$ and $r_{i + 2}$ can visit an inner vertex of a different normal separator path $r ^ {\prime}$, and thus after such replacement the path still avoids $X_{ST}$ or $X$. Therefore $p$ is not the shortest $X_{ST}$-avoiding path or the shortest $X$-avoiding path, which is a contradiction.
\end{proof}

We have already shown that $X_{ST}$ has the \textit{interval property}: indices of inner vertices of separator paths in $X_{ST}$ on $P$ are intervals on $P$. With Lemma \ref{trulyuseful}, the same interval property also applies to $X ^ {+}$. Figure \ref{fig:main5} shows the intervals for $X$ (which in reality is equal to $X_n$) in the example graph we have shown before. To show this, let $r$ be a separator path in $X ^ {+} \backslash X_{ST}$, and $P_x, P_y$ are inner vertices of $r$ such that $|x - y|$ and no $P_z ( x < z < y)$ is an inner vertex of $r$. Then $P_x$ and $P_y$ are strongly $r$-connected. Since $P$ is the shortest path, the edge between $P_x$ and $P_y$ is longer than $P_{x, y}$ and therefore $r$ is not truly useful, violating Lemma \ref{trulyuseful}.

\begin{figure}
    \includegraphics[width=1\textwidth]{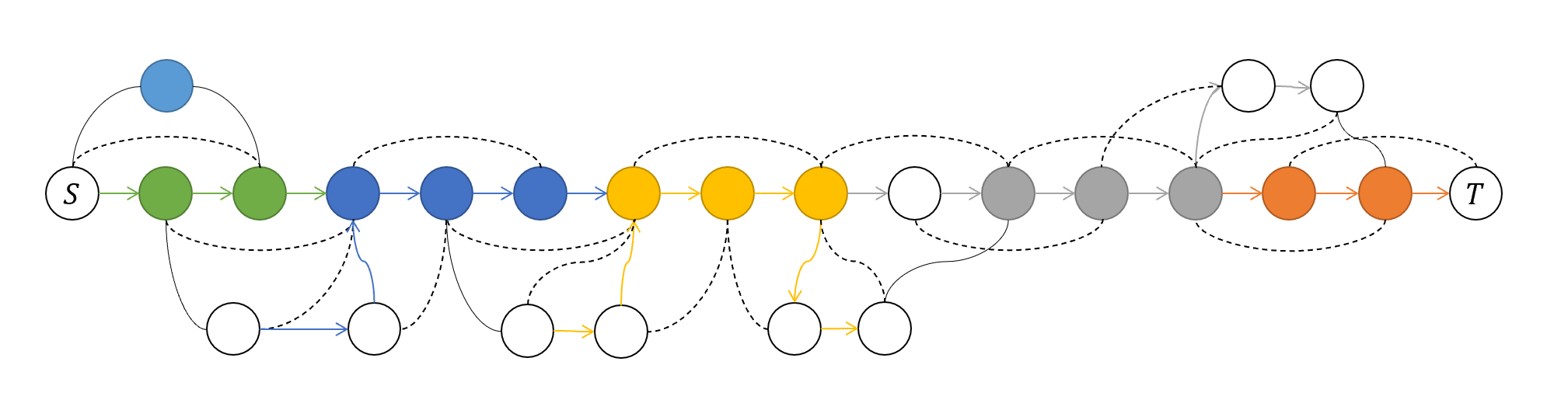}
    \caption{Intervals for $X$}
    \label{fig:main6}
\end{figure}

Recall that $n \in X_n$ is a hypothetical normal separator path $n \in X_n$ and $P_i$ is an inner vertex of $n$. We can let $j$ be the maximum $j < i$ such that one of the following is true:

\begin{enumerate}
    \item $P_j$ is an inner vertex of an $S$-separator path $\alpha$
    \item $P_j \in P_{\textrm{AVOID}}$ and $P_j$ is not an inner vertex of $n$. (Note that $j = 0$ satisfies this condition and therefore $j$ is always well-defined)
    \item $P_j$ is an inner vertex of a separator path $\alpha \in X_{\textrm{EXTRA}}$
\end{enumerate}

We argue that all but the first case are impossible. In the first case, and we let $s = \alpha$ and $a = {\alpha}_{|\alpha| - 3}$. We will show that $a$ is the last inner vertex of $\alpha$ that $P_{\textrm{AVOID}}$ visits.

Similarly we let $k$ be the minimum $k > i$ such that one of the following is true:

\begin{enumerate}
    \item $P_k$ is an inner vertex of a $T$-separator path $\beta$
    \item $P_k \in P_{\textrm{AVOID}}$ and $P_k$ is not an inner vertex of $n$.
    \item $P_k$ is an inner vertex of a separator path $\beta \in X_{\textrm{EXTRA}}$
\end{enumerate}

And we also argue that all but the first case are impossible. In the first case, let $t = \beta$ and $b = {\beta}_{2}$. If for both directions the first case is true, it is straightforward to verify that such $a$ and $b$ satisfies all four conditions stated at the beginning of this proof. 

Due to symmetry, we will only need to examine the direction for $\alpha$. To proceed, we first introduce two lemmas:

\begin{lemma} \label{lemma:stpathseparator}
    For $r \in X ^ {+}$:
    
    If the vertex $u \in r$ on $P$ with the smallest $\textrm{INDEX}(u, P)$ is an endpoint of $r$, and $r_{0, 1}$ is on $P$, then $r$ is an $T$-separator path.
    
    Similarly, if the vertex $u \in r$ on $P$ with the largest $\textrm{INDEX}(u, P)$ is an endpoint of $r$, and $r_{|r| - 2, |r| - 1}$ is on $P$, then $r$ is an $S$-separator path.
\end{lemma}

\begin{proof}
    We prove the first half of the lemma.
    
    Suppose $r$ is not a $T$-separator path. Then $r$ is not an $S$-separator path, and therefore from Lemma \ref{trulyuseful} $r$ is truly useful, let $u$ be the vertex of $r$ on $P$ with the smallest $\textrm{INDEX}(u, P)$. Due to Corollary \ref{noreverse}, if $u$ is an endpoint of $r$, $u$ can not be $r_{|r| - 1}$. Therefore $u = r_0$. If $x$ is the smallest index such that $r_{x, x + 1}$ is not on $P$. If for some $y > x$, $r_y$ is back on $P$ again for the first time. If $\textrm{INDEX}(P, r_y) > \textrm{INDEX}(P, r_x)$ , then $r_x$ and $r_y$ are strongly $r$-connected and $y = x + 2$. Since $P$ is the shortest path between $S$ and $T$, by the definition of truly usefulness this is only possible when $x = |r| - 3$ and $y = |r| - 1$. Thus $r$ is a $T$-separator path. If $\textrm{INDEX}(P, r_y) < \textrm{INDEX}(P, r_x)$, $\textrm{INDEX}(P, r_y) < \textrm{INDEX}(P, r_0)$, which violates the fact that $r_0$ has the smallest index. then  If such $y$ does not exist, then $r_x$ and $T$ are strongly $r$-connected. Therefore $x = |r| - 3$ or $x = |r| - 1$, in both cases $r$ is a $T$-separator path. 

    The reasoning for the second half is similar.
\end{proof}

\begin{lemma} \label{lemma:notst}
    If $x \in X$ is not an $S$-separator path and $P_i$ is an inner vertex of $x$. For any $j > i$ such that $P_j$ is an inner vertex of another normal separator path $y \in X ^ {+}$, no simple path from $S$ to $T$ containing $P_{i, j}$ or $P_{j, i}$ can contain $x$.
    
    If $x \in X$ is not a $T$-separator path and $P_i$ is an inner vertex of $x$. Then for any $j < i$ such that $P_j$ is an inner vertex of another normal separator path $y \in X ^ {+}$, no simple path from $S$ to $T$ containing $P_{j, i}$ or $P_{j, i}$ can contain $x$. 
\end{lemma}

\begin{proof}
    We prove the first half of the lemma.

    Suppose $x \in X$ is not an $S$-separator path and $P_i$ is an inner vertex of $x$. It suffices to prove this for the case where for no $i < i ^ {\prime} < j$, $P_{i ^ {\prime}}$ is an inner vertex of $x$. If a simple path contains $P_{i, j}$ or $P_{j, i}$ and contain $x$. If $P_i = x_k$, consider the edge between $x_k$ and $x_{k + 1}$, if $x_{k, k + 1} \ne P_{i, i + 1}$, then no simple path between $S$ and $T$ can contain $x_{k, k + 1}$ and $P_{i, i + 1}$ at the same time (note that this is not the same as $x_{k + 1, k}$ and $P_{i, i + 1}$.), and no path can contain $P_{i + 1, i}$ without visiting $x_k$ twice. Therefore, $x_{k + 1}$ must be an endpoint and therefore $k + 1 = |x| - 1$. From the interval property we have shown before, $x_{|x| - 1}$ must be the vertex of $x$ on $P$ with the largest index. From Lemma \ref{lemma:stpathseparator}, $x$ is an $S$-separator path, which is a contradiction. 
    
    The other direction is similar. 
\end{proof}

Since $P_{\textrm{AVOID}}$ contains $n$, it visits $P_i$. Let $i ^ {\prime}$ be the smallest index such that $P_{i ^ {\prime}}$ is also inner vertex of $n$. Due to the interval property $P_{i ^ {\prime}, i}$ consists entirely of inner vertices of $n$ and therefore $i ^ {\prime} > j$. If $P_{\textrm{AVOID}}$ also visits $P_j$, consider replacing the part of $P_{\textrm{AVOID}}$ between $P_j$ and $P_{i ^ {\prime}}$ by $P_{i ^ {\prime}, j}$ or $P_{j, i ^ {\prime}}$ depending on the order of the visits. Here we assume $P_j$ is visited before $P_{i ^ {\prime}}$. The reasoning for the other direction is similar. To force a similar contradiction as we have seen before. We argue that the some edge on $P_{j, i ^ {\prime}} \not\subset P_{\textrm{AVOID}}$ --- therefore the replacement changes the path, and that the new path after the replacement still avoids $X$ unless $P_j$ is an inner vertex of some an $S$-separator path $\alpha$. This will preclude the second case. 

Firstly, we argue that $P_{j, i ^ {\prime}} \not\subset P_{\textrm{AVOID}}$. Suppose $P_{j, i ^ {\prime}} \subset P_{\textrm{AVOID}}$. Suppose $P_j$ is the inner vertex of some normal separator path $\alpha$. If $n_{0, 1}$ is not on $P_{j, i ^ {\prime}}$, then since $P_{i ^ {\prime}}$ is an inner vertex of $n$ and $P_j$ is not an inner vertex of $n$, the path $P_{\textrm{AVOID}}$, which contains $P_{j, i}$, can not contain $n$. If $n_{0, i}$ is on $P_{j, i ^ {\prime}}$, we can apply Lemma \ref{lemma:stpathseparator} and argue that $n \in X_{ST}$, which contradicts $X \cap X_n = \emptyset$. 

Obviously, the new path is still simple, if $P_j$ is the inner vertex of some separator path $\alpha$, from lemma \ref{lemma:notst}, unless $\alpha$ is an $S$-separator path, the new path does not contain $\alpha$. Since none of the other elements of $X$ has an inner vertex on $P_{i, j}$, the new path still avoids $X$. 

The only case left is if $P_j$ is the inner vertex of $\alpha \in X_{\textrm{EXTRA}}$ and $P_{\textrm{AVOID}}$ does not visit $P_j$. Suppose $P_j = {\alpha}_k$. From Corollary \ref{corollary:normalvisit} we have $P_{\textrm{AVOID}}$ must visit ${\alpha}_{k - 1}$ and ${\alpha}_{k + 1}$, unless they are endpoints. If ${\alpha}_{k - 1}$ is not an endpoint and is on $P$, then $P_{j - 1} = {\alpha}_{k - 1}$ since otherwise ${\alpha}_{k - 1}$ and ${\alpha}_k$ become strongly ${\alpha}$-connected. Since ${\alpha}_{k - 1}$ is also on $P_{\textrm{AVOID}}$, one can let $j = k - 1$ and one can apply the same reasoning as when $P_{\textrm{AVOID}}$ visits $P_j$ to argue that $\alpha$ is an $S$-separator path, with minor modification, which contradicts $\alpha \in X_{\textrm{EXTRA}}$. Then if ${\alpha}_{k - 1}$ is an inner vertex, ${\alpha}_{k - 1} \notin P$. By the way we find $j$, if ${\alpha}_{k + 1}$ is an inner vertex, ${\alpha}_{k + 1} \notin P$. 

If both ${\alpha}_{k - 1}$ and ${\alpha}_{k + 1}$ are inner vertices, since $\alpha$ is truly useful according to Lemma \ref{trulyuseful}, we can replace the segment of $P_{\textrm{AVOID}}$ from ${\alpha}_{k - 1}$ to ${\alpha}_{k + 1}$ by ${\alpha}_{k - 1, k + 1}$ and make the path $P_{\textrm{AVOID}}$ shorter. Note that the new path might not be $X$-avoiding since it might contain ${\alpha}$, but it visits $P_j$ ($= {\alpha}_k$), and now we can use the reasoning before using replacement to find a even shorter path that does not contain ${\alpha}$ and is $X$-avoiding, which is a contradiction.

If one of ${\alpha}_{k - 1}$ and ${\alpha}_{k + 1}$ is an endpoint. Suppose $k - 1 = 0$. Then ${\alpha}_{k + 1}$ is an inner vertex since $|\alpha| > 3$, and therefore ${\alpha}_{k + 1} \notin P$. From Corollary \ref{corollary:firstvisit}, since ${\alpha}_{k + 1} \notin P$, ${\alpha}_{k - 1} ({\alpha}_0)$ is the first vertex on $\alpha$ that $P$ visits, from Lemma \ref{lemma:stpathseparator} $\alpha$ is a $T$-separator path, which contradicts $\alpha \in X_{\textrm{EXTRA}}$. Since our reasoning is symmetric, we also know that $k + 1$ is not an endpoint.

To show that ${\alpha}_{|\alpha| - 3}$ is the last inner vertex of $\alpha$ that $P_{\textrm{AVOID}}$ visits, we show that $P_{\textrm{AVOID}}$ does not visit the tail of $\alpha$. Suppose it does. Then if ${\alpha}_{|\alpha| - 1}$ is not an inner vertex of $n$, then $P_j$ should have been ${\alpha}_{|\alpha| - 1}$; if ${\alpha}_{|\alpha| - 1}$ is an inner vertex of $n$, one can see that $P_{\textrm{AVOID}}$ can not contain $n$.

\begin{proof} [Proof of Lemma \ref{lemma:noinner}]
    Let $p$ be an $r$-free path between $r_i$ and $r_{i + 2}$. Suppose $p$ an inner vertex of another normal separator path $r ^ {\prime}$. We note the following fact due to Theorem \ref{pathseparatorprop1}: no other vertex on $r$ is $r$-connected to both $r_i$ and $r_{i + 2}$. Let $p ^ {\prime}$ be any path between $S$ and $T$ that contains $r ^ {\prime}$. Let $u$ be the last vertex of $r$ that $p ^ {\prime}$ visits before $r ^ {\prime}$, and if such vertex does not exist, let $u = S$. Let $v$ be the first vertex of $r$ that $p ^ {\prime}$ visits after $r ^ {\prime}$, and if such vertex does not exist, let $v = T$. If $u = r_i$ and $v = r_{i + 2}$, since $r_i$ and $r_{i + 2}$ are adjacent due to Lemma \ref{adjacent}, $(r ^ {\prime})_0$ and $(r ^ {\prime})_{|r ^ {\prime} - 1}$ are strongly $r ^ {\prime}$-connected. Otherwise without loss of generality suppose $u \notin \{r_i, r_{i + 2}\}$. If $u = S$, both $r_i$ and $r_{i + 2}$ are $r$-connected to $r_0$, and otherwise both $r_i$ and $r_{i + 2}$ are $r$-connected to $u$. Neither is possible.
\end{proof}

\subsection{Correctness of $\textrm{AVOID}(X)$} \label{correctavoid}

In this section, we will show the correctness of $\textrm{AVOID}(X)$ based only on the fact that $X$ is a set of normal separator paths. We will ignore all the extra properties of the set $X$ computed in the main procedure we have seen in the previous section. We believe this makes the sub-procedure more general. We will do this in two steps: first, we show that the shortest $X$-avoiding path $p$ from $S$ to $T$ in $G$ corresponds to a path from $(S, 0)$ to $(T, 0)$ in $G_1$. Second, we show that the shortest path $p ^ {\prime}$ from $(S, 0)$ to $(T, 0)$ in $G_1$ corresponds to an $X$-avoiding path in $G$. Obviously these will show the correctness of the path we compute.

\subsubsection{$p$ to $p ^ {\prime}$} \label{ppprime}

If $p$ is the shortest $X$-avoiding path from $S$ to $T$ in $G$, we can construct $p ^ {\prime}$ in this way: start from $(S, 0)$. We move along $p$ and construct $p ^ {\prime}$ alongside. Imagine the two paths being two pointers moving in sync, one in $G$ and the other in $G_1$. Whenever $p$ visit a new vertex, $p ^ {\prime}$ will try to visit the corresponding low vertex, and if an edge to the low vertex does not exist, we visit the high vertex. We show that there will never be a case when neither edge exists. Note that $p$ is a simple path.

We first argue that, for $r \in X$, suppose we are currently at index $i + 1$. $(p ^ {\prime})_{i + 1}$ is the high vertex of an inner vertex $r_j$ of $r$ but $p_i$ is not an inner vertex of $r$, then either $p_i = r_0$, $p_{i + 1} = r_1$, or $p_{i + 1} = r_{|r| - 3}$ and both $r_{|r| - 2}$ and $r_{|r| - 4}$ are visited before $p_i$.

They are only three ways this can happen. Only two of them are possible.

The first way is that $p ^ {\prime}_{i + 1} = (r_1, 1)$ and $p_i = r_0$. Therefore $p_i = r_0$ and $p_{i + 1} = r_1$.

The second way is that $p ^ {\prime}_{i + 1} = (r_{|r| - 3}, 1)$ and $p_i = u$ where $L(r, u) = r_{|r| - 3}$. In this case, $p$ hasn't visited $r_{|r| - 3}$ but has visited $r_{|r| - 1}$. From Corollary \ref{corollary:normalvisit} we can see that both $r_{|r| - 4}$ and $r_{|r| - 2}$ are visited.
 
The third is that $(p ^ {\prime})_i = (v, 1)$ where $(v \notin r)$ and $(p ^ {\prime})_{i + 1} = (r_j, 1)$ where $0 < j < |r| - 1$. From the way we built $G_1$, this will imply the edge $e$ between $v$ and $r_i$ is on some separator path $r ^ {\prime} \in X$ and Theorem \ref{nosharee} implies that such $r ^ {\prime}$ is unique. From Theorem \ref{nosharev} we know $r_i$ is an endpoint of $r ^ {\prime}$. From the way $G_1$ is built we know $e$ is not the tail of $r ^ {\prime}$. Therefore $v = (r ^ {\prime})_0$ and the edge $e$ is from $(r ^ {\prime})_1$ to $(r ^ {\prime})_0$, which means that $(r ^ {\prime})_2$ is visited by $p$. With $(r ^ {\prime})_1$ and $(r ^ {\prime})_2$ both visited before and $p$ currently at $(r ^ {\prime})_0$, $p$ can not get to $T$ while being simple.

Therefore, if $p ^ {\prime}$ visits $(r_i, 1) (0 < i < |r| - 1)$ and was not on the high vertex of an inner vertex of $r$ before this visit, then $r_{i - 1}$ has been visited by $p$. If after the current visit, $p ^ {\prime}$ immediately visits another high vertex of an inner vertex of $r$, it will be $(r_{i + 1}, 1)$, while $r_i$ has been visited by $p$. With this induction we can further deduce that whenever $p ^ {\prime}$ visits a high vertex $(r_k, 1)$, regardless of what we visited before this visit, all vertices on $r_{0, k - 1}$ have been visited by $p$.

If $p$ visits a vertex and $p ^ {\prime}$ is not able to reach any copy of the vertex, the possibilities are: it is trying to visit (some copy of) $r_{i - 1}$ or $r_{i - 2}$ from $(r_i, 1)$; it is trying to visit $r_i$ from $(u, 0)$ where $u \notin r$, $L(u, r_i) = r_i$ and $0 < i < |r| - 1$ and $i \ne |r| - 3$; it is trying to visit $v$ from $(r_i, 1)$ where $0 < i < |r| - 1$, $v \notin u$ and $R(r, v) = r_i$. it is trying to visit $r_{|r| - 1}$ from $(r_{|r| - 2}, 1)$.

If $p ^ {\prime}$ is trying to visit a copy of $r_{i - 1}$ or $r_{i - 2}$ from $(r_i, 1)$, since $r_{i - 1}$ and $r_{i - 2}$ have already been visited by $p$, $p$ will not be simple.

If $p ^ {\prime}$ is trying to visit $v$ from $(r_i, 1)$, where $0 < i < |r| - 1$, $v \notin u$ and $R(r, v) = r_i$. Since $p ^ {\prime}$ is at $(r_i, 1)$, $r_{i - 1}$ has been visited by $p$. If we visit $v$ the next inner vertex of $r$ that we visit on $p$ will be $r_{i - 2}$, $p$ will not be able to reach $T$ without visiting $r_{i - 1}$ or $r_i$ again, but it can not visit do either of them while being simple, which leads to a contradiction.

If $p ^ {\prime}$ is trying to visit $r_{|r| - 1}$ from $(r_{|r| - 2}, 1)$. Consider the last time it visits a high vertex of an inner vertex of $r$. As we showed there are only two ways this can happen. If it got from $u$ to $(r_{|r| - 3}, 1)$ where $L(r, u) = r_{|r| - 3}$ then $r_{|r| - 1}$ has already been visited by $p$, so it will not visit $r_{|r| - 1}$ again. If it visited $(r_1, 1)$ from $r_0$, then by visiting $r_{|r| - 1}$, $p$ will contain $r$, which violates the fact that $p$ avoids $X$.

If $p ^ {\prime}$ is trying to visit $r_i$ from $(u, 0)$ where $u \notin r$, $L(u, r_i) = r_i$ and $0 < i < |r| - 1$ and $i \ne |r| - 3$. From Theorem \ref{pathseparatorprop1}, suppose the last time $p$ gets off $r$ it was at $r_j$. If $0 < i < |r| - 3$, then $j = i$ or $i + 2$. If $i = |r| - 2$, $j = i$. In both cases we violate Corollary \ref{traorder}. 

\subsubsection{$p ^ {\prime}$ to $p$}

If $p ^ {\prime}$ is the shortest $X$-avoiding path from $(S, 0)$ to $(T, 0)$ in $G_1$. We can see that $p ^ {\prime}$ will not visit any $(u, 0)$ before $(u, 1)$ since all paths starting from $(u, 1)$ can be replaced by an path starting from $(u, 0)$, by replacing some high vertices with the corresponding low vertices. For the same reason, we can assume that $p ^ {\prime}$ will not visit $(u, 1)$ from a vertex when it can visit $(u, 0)$ from that vertex, and similar to what we had in the previous part, we can deduce that whenever $p ^ {\prime}$ visits the high vertex of an inner vertex $r_i$ of $r \in X$ from a vertex that does not correspond to an inner vertex of $r$, either it is from $r_0$ to $r_1$ or $i = |r| - 3$ or $i = |r| - 3$ and we visit $r_i$ from $u \notin r$ where $L(r, u) = r_i$.

Let $p$ be the path from $S$ to $T$ in $G$ corresponding to $p ^ {\prime}$ in $G_1$. $p$ obviously does not contain a bad vertex. It suffices to prove that the set of edges on $p$ does not contain the set of edges on any element of $X$ (Note that we are not sure if $p$ is simple yet). Suppose it does for $r$ in $X$. We first show that $p ^ {\prime}$ can not visit any low vertex $(u, 0)$ of $u \notin r$ from any inner vertex $(r_i, 1)$. Suppose it does. Note that since $(r_i, 0)$ hasn't been visited before. One of the edge associated with $r_i$ on $r$ hasn't be visited by the corresponding path $p$. So $p ^ {\prime}$ will have to visit $(r_i, 0)$ again in order to visit the other edge. By the way we built $G_1$, $L(r, u) = r_i$ and the next vertex on $r$ that $p$ visits must be $r_{i + 2}$, and the corresponding visit on $p ^ {\prime}$ is either $(r_{i + 2}, 0)$ or $(r_{i + 2}, 1)$. Between them we can rule out $(r_{i + 2}, 1)$ since $p ^ {\prime}$ visits $r(r_{i + 2}, 1)$ from a vertex that does not correspond to an inner vertex of $r$, and we can not be in either of the two cases this can happen. Since $p ^ {\prime}$ visited $(r_{i + 2}, 0)$ before visiting $(r_i, 0)$, $p$ will not visit $r_{i + 2}$ again after visiting $r_i$. Thus one can see that after visiting $r_i$, there will be a segment of $p$ corresponding to a path from $r_{i + 1}$ to $r_{i + 3}$ (or to $T$ if $i + 2 = |r| - 1$) such that no vertices on the segment except its endpoints are on $r$. The suffix of $p$ after the start of this segment corresponds to a path from $(r_{i + 1}, 1)$ in $G_1$, and we can replace the part of $p ^ {\prime}$ from $(r_{i + 1}, 0)$ with this suffix from $(r_{i + 1}, 1)$. This replacement preserves $p ^ {\prime}$'s length as well as its correspondence with $p$. Note that to contain all the edges on $r$, $p$ needs to visit at least three edges associated with $r_{i + 1}$. Therefore $p$ visits $r_{i + 1}$ at least twice. Therefore $p ^ {\prime}$ after the length-preserving replacement must visit either $(r_{i + 1}, 1)$ or $(r_{i + 1}, 0)$ again after visiting $(r_{i + 1}, 1)$, and in either case we can do a second replacement and make $p ^ {\prime}$ a shorter path between $S$ and $T$ in $G_1$, which is impossible since $p ^ {\prime}$ is the shortest.

Now that we have shown that $p ^ {\prime}$ can not visit any low vertex $(u, 0)$ of $u \notin r$ from any inner vertex $(r_i, 1)$. Since $p$ contains all the edges on $r$, $p$ must visit the edge between $r_0$ and $r_1$. If $p$ never visits $r_1$ from $r_0$ but has visited $r_0$ from $r_1$, then by the time $p$ visits $r_1$, $p ^ {\prime}$ must have visited $(r_1, 0)$ and $(r_2, 0)$ before. Since $p ^ {\prime}$ can not visit any copy of $r_1$ or $r_2$ again, $p ^ {\prime}$ can not reach $(T, 0)$. Suppose $p$ visits $r_1$ from $r_0$. After visiting the edge $p ^ {\prime}$ will be at $(r_1, 1)$. Since the edge from $(r_{i + 1}, 1)$ to $(r_i, 0)$ for $i > 0$ or $(r_{i - 2}, 0)$ for $i > 1$ does not exist, we will have to take the route $(r_1, 1) \rightarrow (r_2, 1) \rightarrow (r_3, 1) \cdots$. Eventually, we will be forced to go from $(r_{|r| - 2}, 1)$ to $(r_{|r| - 1}, 1)$, which is an edge that does not exist in $G_1$, and we have a contradiction. 

\section{Tie Breaking} \label{tie}

Throughout this paper we have assumed that any two different paths between the same pair of vertices (ordered pair) in the input graph have different lengths. However, in reality, this is not already true. There can be ties where two different paths between some pair of vertices have the same lengths. One can see that if we break ties arbitrarily, the final path we get might be a separating path equal in length to the real answer. In Figure \ref{fig:tie}, numbers beside the edges are their lengths. If both the path $P$ and the path $P_0$ go through the blue part. The final path we compute may go through the green part, which means that the path is separating. Fortunately, there are ways we can break ties.

\begin{figure}
    \includegraphics[width=1\textwidth]{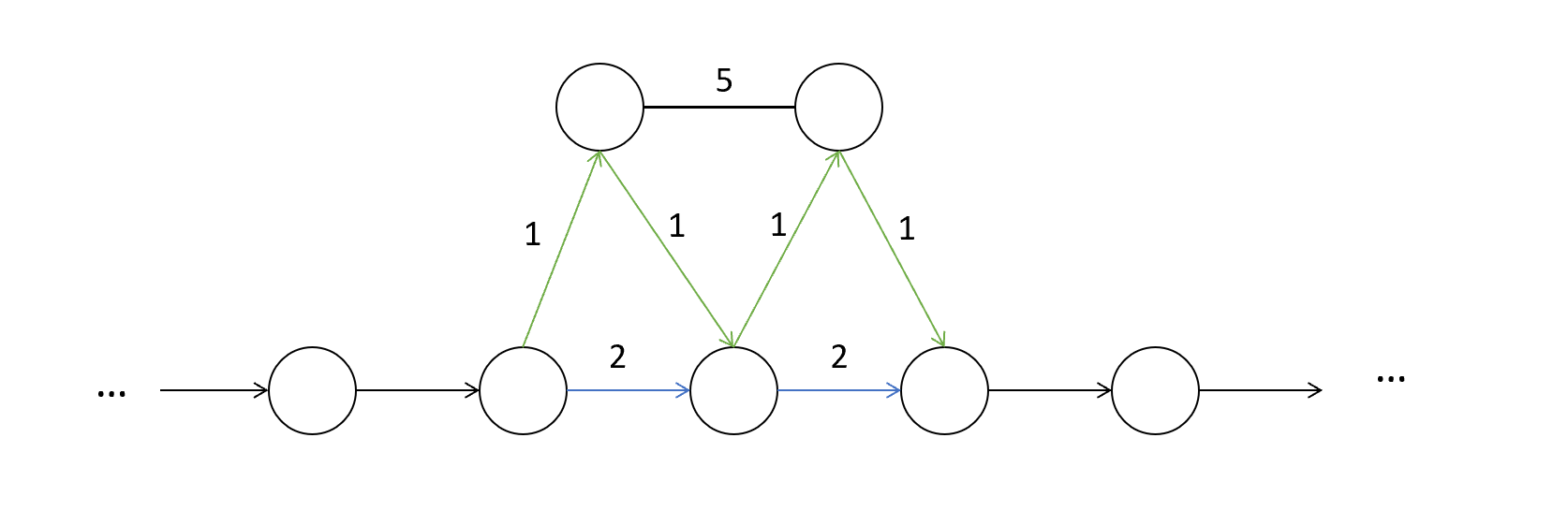}
    \caption{Tie breaking issues}
    \label{fig:tie}
\end{figure}

A tie-breaking scheme is called \textit{consistent} if for any two paths connecting the same pair of vertices with the same length, the scheme always favors one after another. Many implementations of the Dijkstra's Algorithm achieves uses this time breaking scheme since the order in which we iterate through all the edges associated with a particular vertex is usually fixed. If we use Dijkstra's Algorithm in our main algorithm, for most of the common ways one can store an undirected graph on a machine using the RAM model, a consistent tie-breaking scheme is easy to achieve. For example, if the graph is stored by using linked lists, for all the auxiliary graphs involved in our algorithm, it suffices to store the linked lists in a way such that if at a particular vertex the edge $e_0$ before the edge $e_1$, then at the corresponding vertex in the auxiliary graph, all edges corresponding to $e_0$ comes before all edges corresponding to $e_1$.

\section{NP-hardness proof for general graphs}

We will show that the problem of deciding the existence of non-separating st-paths is NP-hard on general graphs.

We will reduce the famous \textit{3-SAT} problem to the problem of deciding the existence of non-separating st-paths. Here we adopt the notational framework in \cite{welzl}. In the 3-SAT problem, we are given a 3-CNF formula, which is a conjunction of $n$ clauses and $m$ variables, each with exactly 3 literals, and we want to decide whether a satisfying assignment exists. The well-known Cook-Levin Theorem \cite{cook, levin} states that this problem is NP-Hard. We show that we can solve 3-SAT by deciding the existence of non-separating st-paths on a graph with $O(n)$ (assuming that $m = O(n)$) vertices and edges, which will then prove Theorem \ref{nphard}.

Let the $i$-th variable be $x_i$. Let $k_i$ be the amount of times its complement $\overline{x_i}$ appears in the formula. Let $\overline{k_i}$ be the amount of times $x_i$ appears in the formula. To build the graph, we first create the vertex $S$. For variable $x_i$ in order, we create $k_i + \overline{k_i} + 1$ nodes: $a_{i, 0 \cdots k_i - 1}, \overline{a_{i, 0 \cdots \overline{k_i} - 1}}, b_i$. Let $t = b_{i - 1}$ ($t = S$ if $i = 0$). Make edges between $t$ and $a_{i, 0}$ (if it exists), between $t$ and $\overline{a_{i, 0}}$ (if it exists), between $a_{i, j}$ and $a_{i, j + 1}$ (for all $j$ such that both vertices exist), between $\overline{a_{i, j}}$ and $\overline{a_{i, j + 1}}$ (for all $j$ such that both vertices exist), between $a_{i, k_i - 1}$ (if it exists) and $b_i$, between $\overline{a_{i, \overline{k_i} - 1}}$ and $b_i$ (if it exists). If $k_i = 0$ or $\overline{k_i} = 0$, make an edge between $t$ and $b_i$. Finally, let $T = b_{m - 1}$.

We use the following shorthand: To make a \textit{fat edge} between $u$ and $v$, we make a dummy node $w$, and then make edges between $u$ and $w$ and between $w$ and $v$. A fat edge can be treated as a single edge that can not be traversed by the path, since any path that goes through $u \rightarrow w \rightarrow v$ makes $w$ disconnected from the rest of the graph. For the $k$-th clause, create a vertex $c_k$. If the clause contains literal $x_i$, and is the $j$-th clause that contains that literal, make a fat edge between $\overline{a_{i, j}}$ and $c_k$. If the clause contains literal $\overline{x_i}$, and is the $j$-th clause that contains that literal, make a fat edge between $a_{i, j}$ and $c_k$.

One can see that the 3-CNF formula is satisfiable if and only if the there exists a non-separating path from $S$ to $T$ in the graph. A solution of the formula corresponds to the path in this way: for the $i$-th variable, if $x_i$ is true, the path goes through $a_{i, *}$. Otherwise, the path goes through $\overline{a_{i, *}}$. (The path goes through the edge between $b_{i - 1}$ (or $S$) and $b_i$ if the corresponding $k_i$ or $\overline{k_i}$ is $0$) In this way, it is easy to see that for a non-satisfying assignment, any clause that is violated will be disconnected from the rest of the graph.

The following is an example for the formula: $(x_0 \vee \overline{x_1} \vee x_2) \wedge (\overline{x_0} \vee x_1 \vee x_2)$:

\begin{figure}
    \includegraphics[width=1\textwidth]{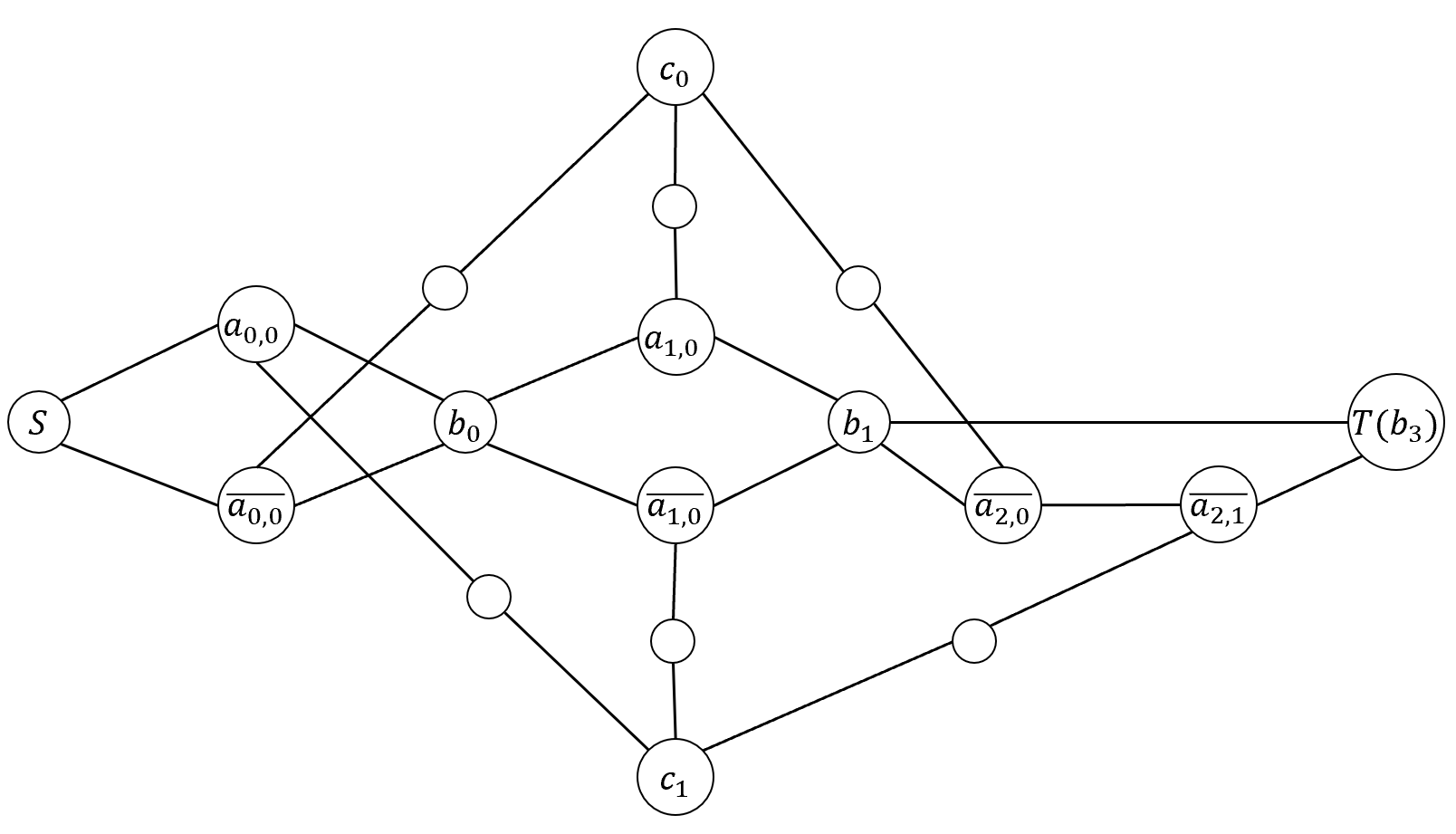}
    \caption{The graph.}
    \label{fig:nph1}
\end{figure}

Figure \ref{fig:nph1} shows the graph we build for the forumula.

\begin{figure}
    \includegraphics[width=1\textwidth]{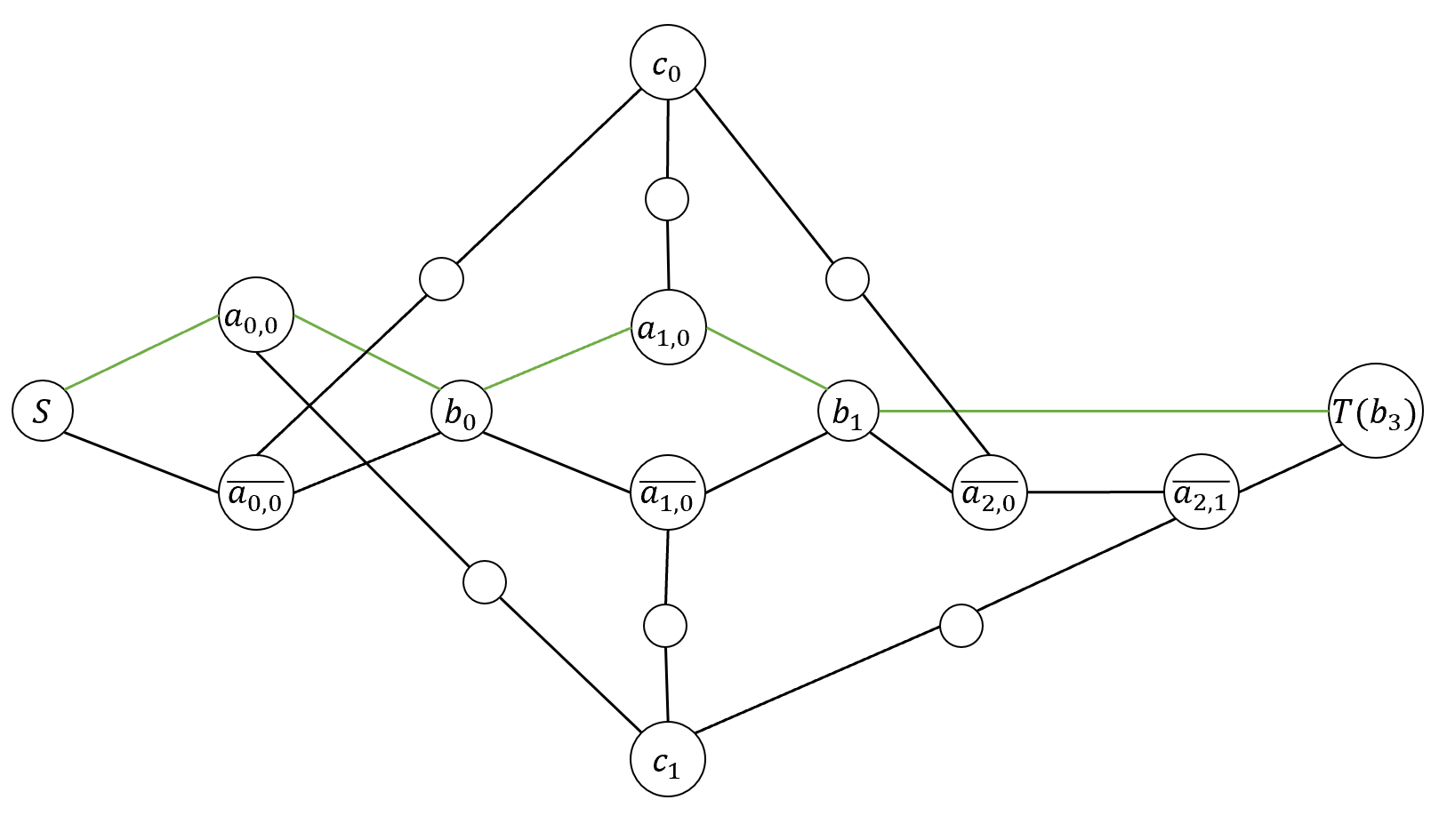}
    \caption{A non-separating path}
    \label{fig:nph2}
\end{figure}

Figure \ref{fig:nph2} shows the non-separating path corresponding to assignment $\{x_0 = \textrm{TRUE}, x_1 = \textrm{TRUE}, x_2 = \textrm{TRUE}\}$, which is a satisfying assignment.

\begin{figure}
    \includegraphics[width=1\textwidth]{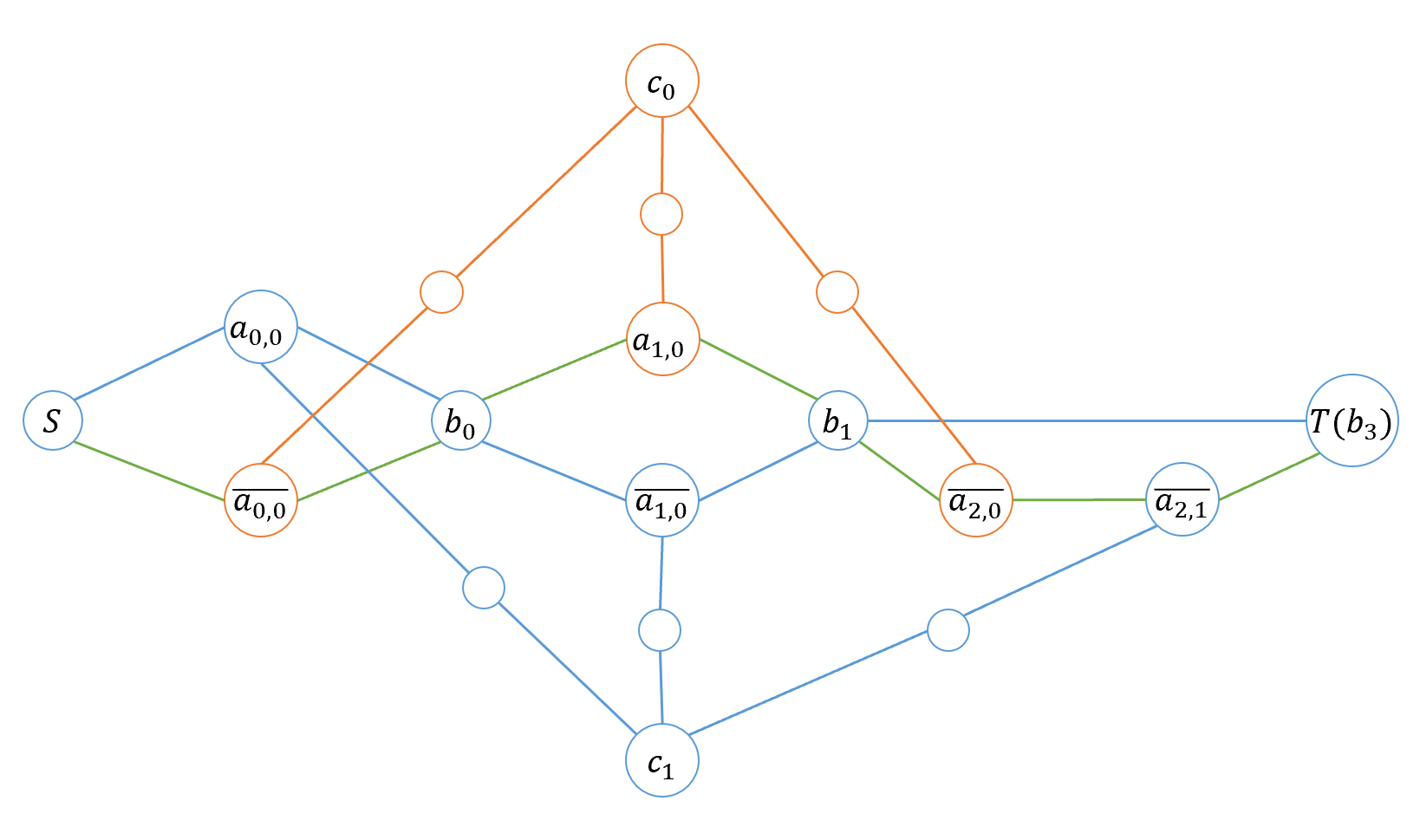}
    \caption{Not a non-separating path}
    \label{fig:nph3}
\end{figure}

Figure \ref{fig:nph3} shows the path corresponding to assignment $\{x_0 = \textrm{FALSE}, x_1 = \textrm{TRUE}, x_2 = \textrm{FALSE}\}$ in green, which is not a satisfying assignment as clause $(x_0 \vee \overline{x_1} \vee x_2) $ is violated. The path is not a non-separating path. After removing all the edges on the path, the connected component for the violated clause $c_0$ in orange is disconnected from the rest of the graph in light blue.

\bibliographystyle{unsrt}
\bibliography{main}




\end{document}